\documentclass[12pt]{article}

\usepackage{microtype}
\usepackage[margin=1in]{geometry}

\usepackage{multicol}
\usepackage[normalem]{ulem}
\usepackage{amsthm}
\usepackage{amsmath}
\usepackage{amsfonts}
\usepackage[english]{babel}
\usepackage[utf8]{inputenc}
\usepackage{amssymb}
\usepackage{url,hyperref}
\usepackage{wrapfig}
\usepackage{graphicx}
\usepackage{float}
\usepackage{subcaption}

\usepackage[ruled,vlined]{algorithm2e}

\usepackage{nccmath}
\newcounter{algoline}
\newcommand\Numberline{\refstepcounter{algoline}\nlset{\thealgoline}}
\usepackage{subcaption}
\usepackage{enumitem}
\usepackage{xcolor}

\newtheorem{theorem}{Theorem}[section]
\newtheorem{definition}[theorem]{Definition}
\newtheorem{lemma}[theorem]{Lemma}
\newtheorem{claim}[theorem]{Claim}
\newtheorem{corollary}[theorem]{Corollary}
\newtheorem{proposition}[theorem]{Proposition}

\definecolor{darkred}{rgb}{1, 0.1, 0.3}
\definecolor{darkblue}{rgb}{0.1, 0.1, 1}
\definecolor{red}{rgb}{1, 0, 0}
\definecolor{darkgreen}{rgb}{0,0.6,0.5}

\newcommand{\Idomain} 	{\Omega}
\newcommand{\myw}		{{\omega}}
\newcommand{\thickenG}   {{G^\myw}} 
\newcommand{\redR}		{\mathrm{cl}(\overline{G^\myw})} 
\newcommand{\redregion} {{outside-region\xspace}}
\newcommand{\greenregion}  {{neighborhood\xspace}}
\newcommand{\etal}		{{et al.\xspace{}}}

\newcommand{\Ghat}		{{\widehat{G}}}
\newcommand{\WWLAlg}	{{\sf MorseRecon\xspace}}

\newcommand{\Perall}		{{P}} 
\newcommand{\GV}		{{M}}

\newcommand{\dimension}	{{dimension}}
\newcommand{\validcancel}		{{cancellable\xspace}}
\newcommand{\LowSt}		{{\mathrm{LowSt}}}
\newcommand{\myf}		{{\bar{f}}}
\newcommand{\trueG}		{{\mathrm{G}}}

\newcommand{\pers}	{\mathrm{pers}}
\newcommand\K{K}
\newcommand{\mytree}		{{\mathcal{T}}}
\newcommand{\myv}			{{\mathsf{v}}}
\newcommand{\afunc}			{{g_\rho}}
\newcommand{\SimpVecF}		{{\sf PerSimpVF\xspace}}

\newcommand{\SimpVecFVtwo}	{{\sf PerSimpTree\xspace}}
\newcommand{\CollectOutput}	{{\sf CollectOutputG\xspace}}
\newcommand{\CollectOutputVtwo}	{{\sf Treebased-OutputG\xspace}}
\newcommand{\WWLAlgSimp}	{{\sf MorseReconSimp\xspace}}

\newcommand{\mysink}		{{\mathrm{si}}}

\title{Graph Reconstruction by Discrete Morse Theory }
\author{Tamal K. Dey\thanks{Department of Computer Science and Engineering, The Ohio State University. \texttt{tamaldey, yusu@cse.ohio-state.edu, wang.6195@osu.edu}}, Jiayuan Wang\footnotemark[1], Yusu Wang\footnotemark[1]}

\date{}

\begin{document}

 \maketitle

\begin{abstract}
Recovering hidden graph-like structures from potentially noisy data is a fundamental task in modern data analysis. Recently, a persistence-guided discrete Morse-based framework to extract a geometric graph from low-dimensional data has become popular. 
However, to date, there is very limited theoretical understanding of this framework in terms of graph reconstruction. This paper makes a first step towards closing this gap. Specifically, first, leveraging existing theoretical understanding of persistence-guided discrete Morse cancellation, we provide a simplified version of the existing discrete Morse-based graph reconstruction algorithm. We then introduce a simple and natural noise model and show that the aforementioned framework can correctly reconstruct a graph under this noise model, in the sense that it has the same loop structure as the hidden ground-truth graph, and is also geometrically close. We also provide some experimental results for our simplified graph-reconstruction algorithm.  
 \end{abstract}

\section{Introduction}

Recovering hidden structures from potentially noisy data is a fundamental task in modern data analysis. A particular type of structure often of interest is the geometric graph-like structure. 
For example, given a collection of GPS trajectories, recovering the hidden road network can be modeled as reconstructing a geometric graph embedded in the plane. Given the simulated density field of dark matters in universe, finding the hidden filamentary structures is essentially a problem of geometric graph reconstruction.

Different approaches have been developed for reconstructing a curve or a metric graph from input data. For example, in computer graphics, much work have been done in extracting 1D skeleton of geometric models using the medial axis or Reeb graphs \cite{dey2006defining,yan2016erosion,liu2011extended,ge2011data,natali2011graph,biasotti2008reeb}.  
In computer vision and machine learning, a series of work has been developed based on the concept of \emph{principal curves}, originally proposed by Hastie and Steutzle \cite{hastie1989principal}. 
Extensions to graphs include the work in \cite{kegl2002piecewise} for 2D images and in \cite{ozertem2011locally} for high dimensional point data. 

In general, there is little theoretical guarantees for most approaches developed in practice to extract hidden graphs. One exception is some recent work in computational topology: Aanijaneya et al. \cite{aanjaneya2012metric}
proposed the first algorithm to approximate a metric graph from an input metric space with guarantees. The authors of \cite{chazal2015gromov,ge2011data} used Reeb-like structures to approximate a hidden (metric) graph with some theoretical guarantees. These work however only handles (Gromov-)Hausdorff-type of noise. When input points are embedded in an ambient space, they requires the input points to lie within a small tubular neighborhood of the hidden graph. Empirically, these methods do not seem to be effective when the input contains ambient noise allowing some faraway points from the hidden graph. 
 
Recently, a discrete Morse-based framework for recovering hidden structures was proposed and studied \cite{DRS15,GDN07,RWS11}. 
This line of work computes and simplifies a discrete analog of
(un)stable manifolds of a Morse function by using the (Forman's) discrete  Morse theory coupled with persistent homology for 2D or 3D volumetric data.
One of the main issues in such simplification is the inherent obstructions that may occur for cancelling critical pairs. The authors of \cite{RWS11} 
suggest sidestepping this and consider
a combinatorial representation of critical pairs for further processing. The authors in \cite{DRS15} identify a restricted set of pairs called ``cancellable close pairs'' which are guaranteed to admit cancellation.
This framework has been applied to, for example, extracting filament structures from simulated dark matter density fields \cite{2011MNRAS} and reconstructing road networks from GPS traces \cite{WWL15}.

This persistence-guided discrete Morse-based framework has shown to be very effective in recovering a hidden geometric graph from (non-Hausdorff type) noise and non-homogeneous data. The method draws upon the global topological structure hidden in the input scalar field and thus is particularly effective at identifying junction nodes which has been a challenge for previous approaches that rely mostly on local information.
However, to date, theoretical understanding of such a framework remains limited. 
Simplification of a discrete Morse gradient vector field using persistence has been studied before.
For example, the work of \cite{DRS15} clarifies the connection between persistence-pairing and the simplification of \emph{discrete Morse chain complex} (which is closely related, but different from the cancellation in the discrete gradient vector field) for 2D and 3D domains. 
Bauer \etal{} \cite{Bauer2012} obtain several results on persistence guided discrete Morse simplification for combinatorial surfaces. 
The simplification of vertex-edge persistence pairing used in \cite{Bauer2012} has also been observed in \cite{attali2009persistence} independently for simplifying Morse functions on surfaces.
Leveraging these existing developments, we aim to provide a theoretical understanding of a persistence-guided discrete Morse based approach to reconstruct a hidden geometric graph. 

\vspace{-0.2in}
\subparagraph*{Main contributions and organization of paper.}
In Section \ref{sec:alg}, we start with one version of the existing persistence-guided discrete Morse-based graph reconstruction algorithm (as employed in \cite{2011MNRAS,WWL15,newpaper}). 
We show that this algorithm can be significantly simplified while still yielding the same output. 
To establish the theoretical guarantee of the reconstruction algorithm, we introduce a simple yet natural noise model in Section \ref{sec:noisemodel}. Intuitively, this noise model assumes that we are given an input density field $\rho: \mathbb{R}^d \to \mathbb{R}$ where densities are significantly higher within a small neighborhood around a hidden graph than outside it. 
Under this noise model, we show that the reconstructed graph has the same loop structure as the hidden graph, and is also geometrically close to it; the technical details are in Sections \ref{sec:general} and \ref{sec:2D} for the general case and the 2-dimensional case (with additional guarantees), respectively. 

While our noise model is simple, our theoretical guarantees are first of a kind developed for a discrete Morse-based approach applied to graph reconstruction. In fact, prior to this, it was not clear whether a discrete Morse based approach can recover a graph even if there is no noise, that is, the density function has positive values {\it only} on the hidden graph. 
For our specific noise model, it may be possible to develop thresholding strategies perhaps with theoretical guarantees. However, previous work (e.g, \cite{2011MNRAS,WWL15}) have shown that discrete Morse approach succeeds in many cases handling non-homogeneous data where thresholding fails. 
 We have implemented the proposed simplified algorithm and tested it on several data sets, which generally gives a speed-up of at least a factor of 2 over a state-of-the-art approach.  We present more discussions and experimental results  in the appendix.  

\section{Preliminaries}

\subsection{Morse theory}

For simplicity, we consider only a smooth function $f: \mathbb{R}^d \to \mathbb{R}$. 
See \cite{EH10,Morse_Theory} for more general discussions. 

For a point $p\in \mathbb{R}^d$, the gradient vector of $f$ at a point $p$ is $\nabla f(p) = -[\frac{\partial f}{\partial x_1} \cdots  \frac{\partial f}{\partial x_d}]^T$, which represents the steepest descending direction of $f$ at $p$,  with its magnitude being the rate of change. 
An integral line of $f$ is a path $\pi: (0,1)\rightarrow \mathbb{R}^d$ such that the tangent vector at each point $p$ of this path equals $\nabla f(p)$, which is intuitively a flow line following the steepest descending direction at any point. A point $p\in \mathbb{R}^d$ is \emph{critical} if its gradient vector vanishes, i.e, $\nabla f(p) = [0 \cdots 0]^T$. 
A \emph{maximal} integral line necessarily ``starts'' and ``ends'' at critical points of $f$; that is, $\lim_{t\to 0} \pi(t) = p$ with $\nabla f(p) = [0 \cdots 0]^T$, and $\lim_{t\to 1} \pi(t) = q$ with $\nabla f(q) = [0 \cdots 0]^T$. See Figure \ref{fig:cripts} where we show the graph of a function $f: \mathbb{R}^2 \to \mathbb{R}$, and there is an integral line from $p'$ to the minimum $v_1$. 

For a critical point $p$, the union of $p$ and all the points from integral lines flowing into $p$ is referred to as the \emph{stable manifold of $p$}. Similarly, for a critical point $q$, the union of $q$ and all the points on integral lines starting from $q$ is called the \emph{unstable manifold of $q$}.
The stable manifold of a minimum $p$ intuitively corresponds the basin/valley around $p$ in the terrain of $f$. 
The 1-stable manifolds of index ($d-1$) saddles consist of pieces of curves connecting ($d-1$)-saddles to maxima -- These curves intuitively capture ``mountain ridges'' of the terrain (graph of the function $f$); see Figure \ref{fig:cripts} for an example. 
Symmetrically, the unstable manifold of a maximum $q$ corresponds to the mountain around $q$.
The 1-unstable manifolds consist of a collection of curves connecting $1$-saddles to minima, corresponding intuitively to the ``valley ridges''. 

In this paper, we focus on a graph-reconstruction framework using Morse-theory (as in e.g, \cite{GDN07,DRS15,2011MNRAS,WWL15}). 
Intuitively, the 1-stable manifolds of saddles (mountain ridges) of the density field $\rho$ are used to capture the hidden graphs.
To implement such an idea in practice, the \emph{discrete Morse theory} is used for robustness and simplicity contributed by its combinatorial nature; and a simplification scheme guided by the persistence pairings is employed to remove noise. 
Below, we introduce some necessary background notions in these topics. 

\begin{figure}[tbph]
\captionsetup[subfigure]{justification=centering}
    \centering
    \begin{subfigure}[b]{0.2\textheight}

\includegraphics[width=\textwidth]{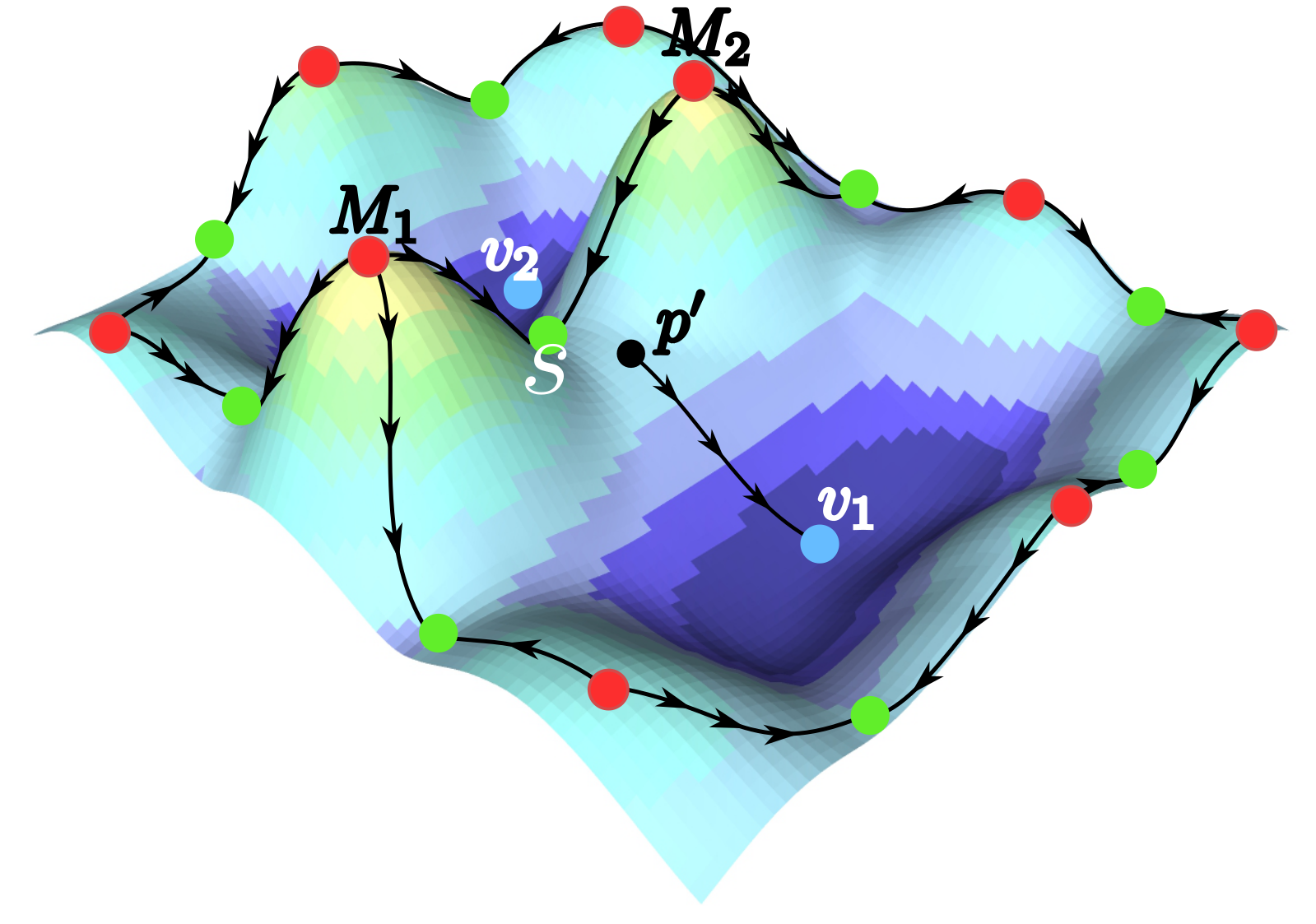}
        \caption{}
        \label{fig:cripts}
    \end{subfigure}
     \begin{subfigure}[b]{0.2\textwidth}
        \includegraphics[width=\textwidth]{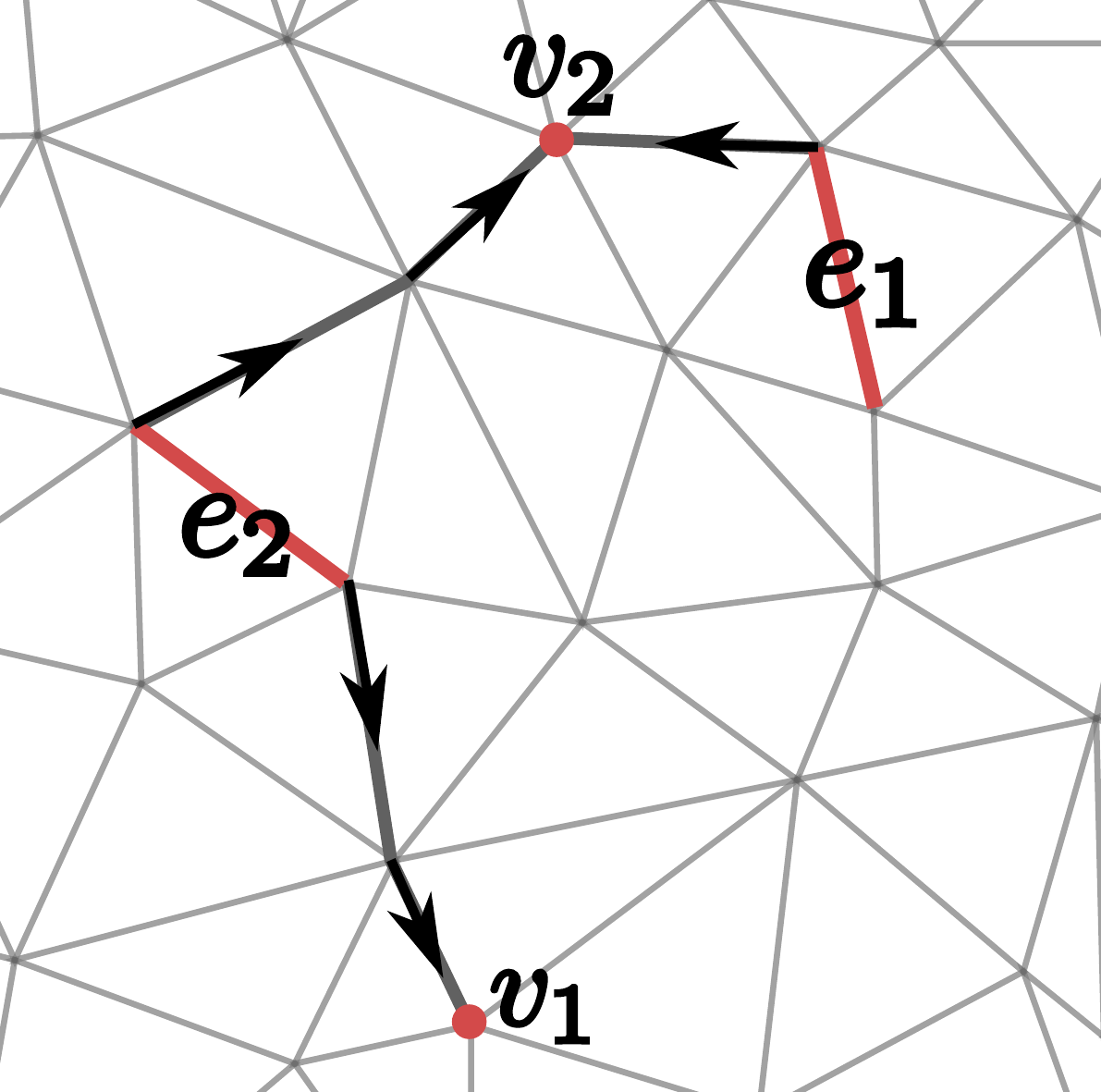}
        \caption{}
        \label{fig:bfcancel}
    \end{subfigure}
    \begin{subfigure}[b]{0.2\textwidth}
        \includegraphics[width=\textwidth]{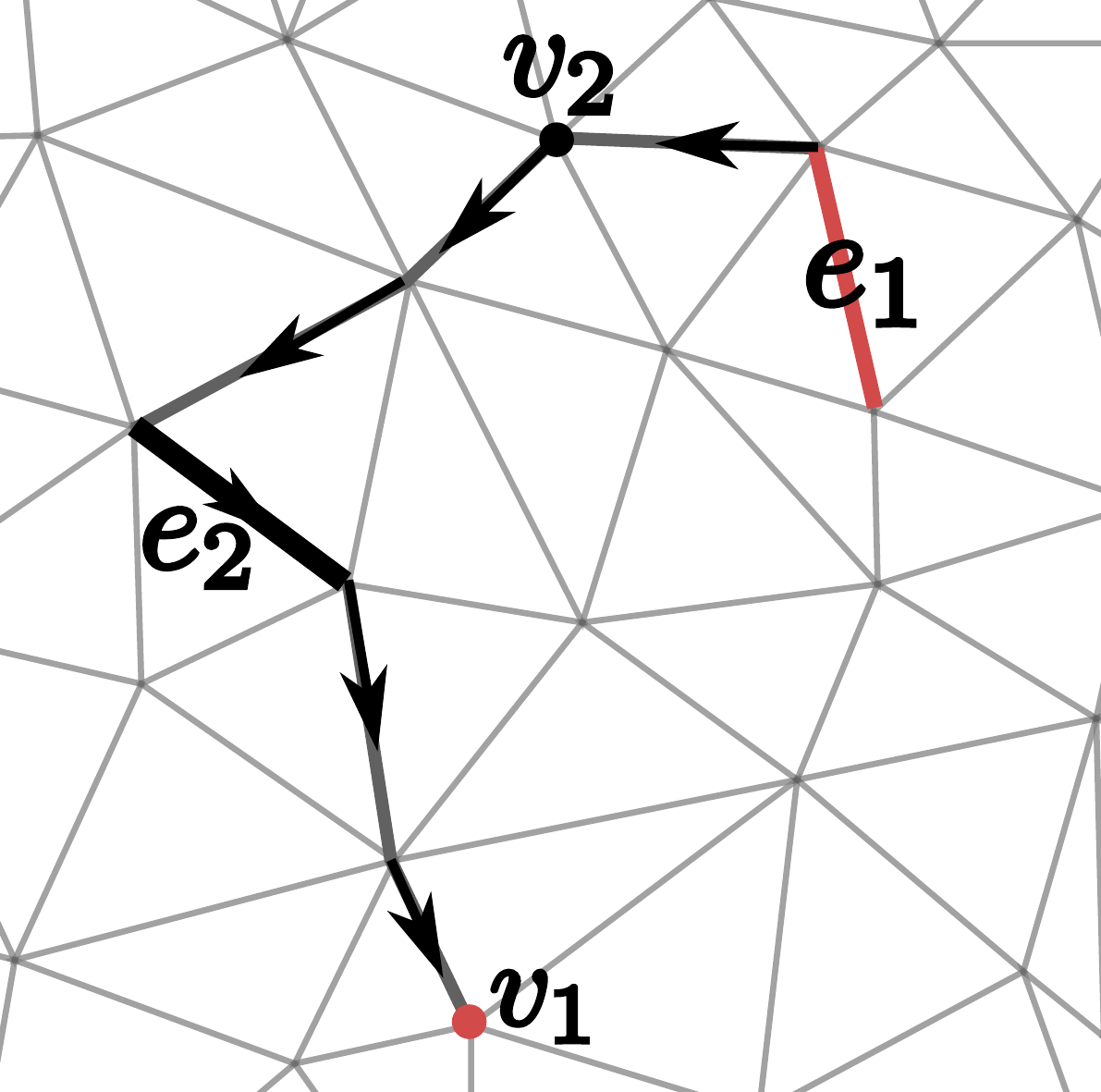}
        \caption{}
        \label{fig:afcancel}
    \end{subfigure}
    \begin{subfigure}[b]{0.2\textwidth}
         \includegraphics[width=\textwidth]{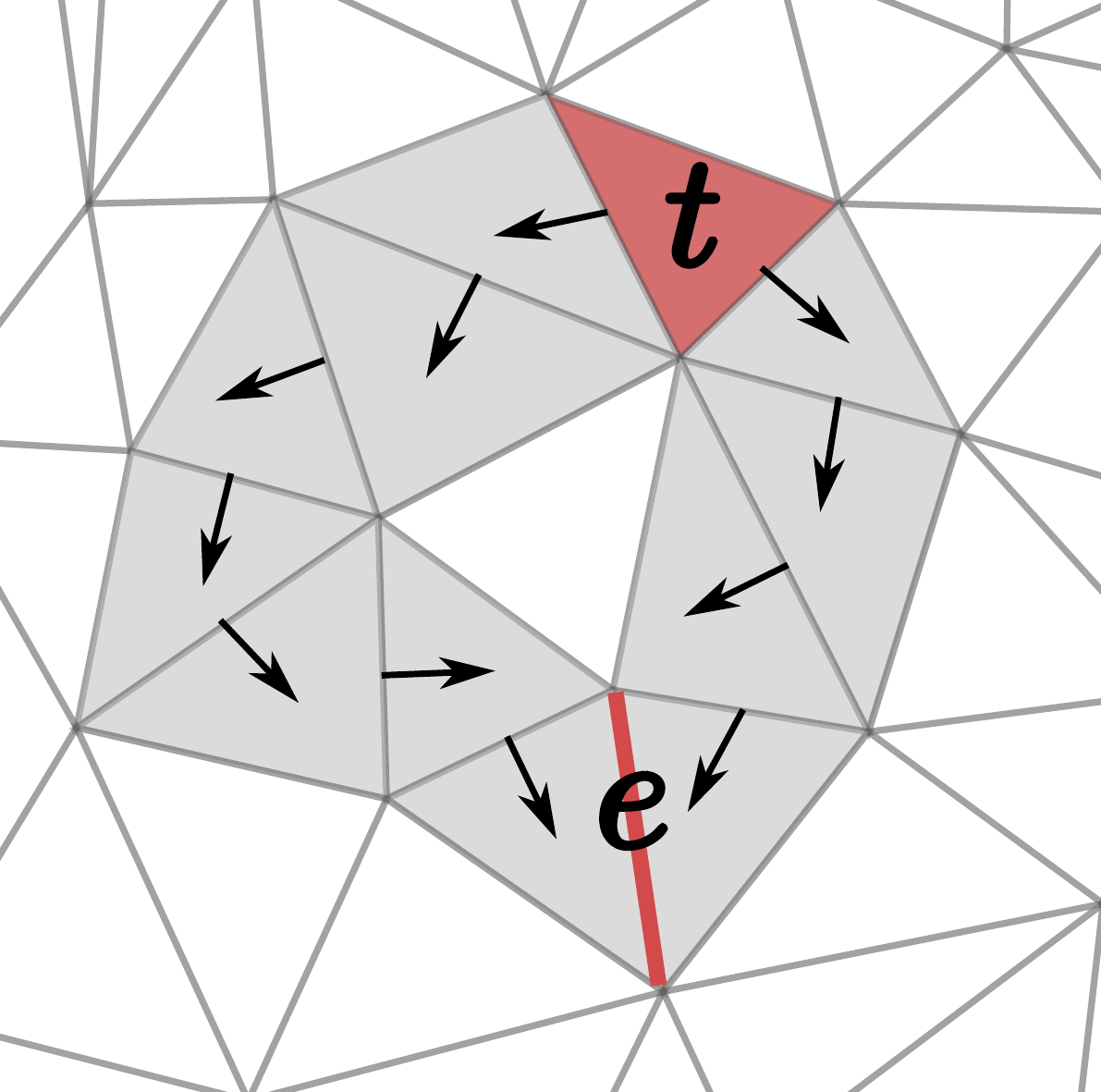}
        \caption{}
        \label{fig:cantCancel}
    \end{subfigure}

 \vspace*{-0.1in}   \caption{(a) $M_1$ and $M_2$ are maxima (red dots), $v_1$ and $v_2$ are minima (blue dots), $s$ is a saddle (green dots) with its stable manifolds flowing to it from $M_1$ and $M_2$. If we put a drop of water at $p'$ it will flow to $v_1$. If we put it on the other side of the mountain ridge it will flow to minimum $v_2$. 
(b) Before cancellation of pair $\langle v_2, e_2\rangle$. (c) After cancellation, the path from $e_2$ to $v_2$ is inverted, giving rise to a gradient path from $e_1$ to $v_1$, making $\langle v_1, e_1\rangle$ now potentially \validcancel. (d) An edge-triangle pair $\langle e, t\rangle$ which is not \validcancel{} as there are two gradient paths between them.
\label{fig:defs} }

\end{figure}

\subsection{Discrete Morse theory}
First we briefly describe some notions from discrete Morse theory (originally introduced by Forman \cite{forman}) in the simplicial setting. 

A $k$-simplex $\tau = \{p_0, \ldots, p_k\}$ is the convex hull of $k+1$ affinely independent points; 
$k$ is called the \emph{dimension} of $\tau$. A \emph{face $\sigma$ of $\tau$} is a simplex spanned by a proper subset of vertices of $\tau$; $\sigma$ is a \emph{facet} of the $k$-simplex $\tau$, denoted by $\sigma < \tau$, if its dimension is $k-1$. 

Suppose we are given a simplicial complex $\K$ which is simply a collection of simplices and all their faces so that if two simplices intersect, they do so in a common face. 
A \emph{discrete (gradient) vector} is a pair of simplices $(\sigma, \tau)$ such that $ \sigma < \tau$. 
A \emph{Morse pairing} in $\K$ is a collection of discrete vectors $M(\K)=\{(\sigma,\tau)\}$ where each simplex appears in \emph{at most} one pair; simplices that are not in any pair are called \emph{critical}. 

Given a Morse pairing $M(\K)$, a \emph{V-path} is a sequence 
$\tau_0,\sigma_1,\tau_1,\ldots,\sigma_\ell, \tau_\ell, \sigma_{\ell+1}, $
where $(\sigma_i, \tau_i)\in M(\K)$ for every $i=1,\ldots,\ell$, and each $\sigma_{i+1}$ is a facet of $\tau_i$ for each $i=0,\ldots,\ell$. If $\ell=0$, the V-path is \emph{trivial}. This V-path is \emph{cyclic} if $\ell> 0$ and $(\sigma_{\ell+1},\tau_0)\in M(\K)$; otherwise, it is \emph{acyclic} in which case we call this V-path a \emph{gradient path}.
We say that a gradient path is a vertex-edge gradient path if $\dimension(\sigma_i) = 0$, implying that $\dimension(\tau_i) = 1$. 
Similarly, it is a edge-triangle gradient path if $\dimension(\sigma_i) = 1$. 
A Morse pairing $M(\K)$ becomes a \emph{discrete gradient vector field} (or equivalently a \emph{gradient Morse pairing}) if there is no cyclic V-path induced by $M(\K)$. 

Intuitively, given a discrete gradient vector field $M(\K)$, a gradient path $\tau_0,\sigma_1,\ldots,\tau_\ell,\sigma_{\ell+1}$ is the analog of an integral line in the smooth setting. But different from the smooth setting, a maximal gradient path may not start or end at critical simplices. 
However, those that do (i.e, when $\tau_0$ and $\sigma_{k+1}$ are critical simplices) are analogous to maximal integral line in the smooth setting which ``start'' and ``end'' at critical points, and for convenience one can think of \emph{critical $k$-simplices} in the discrete Morse setting as \emph{index-$k$ critical points} in the smooth setting. For example, for a function on $\mathbb{R}^2$, critical 0-, 1- and 2-simplices in the discrete Morse setting correspond to minima, saddles and maxima in the smooth setting, respectively. 

For a critical edge $e$, we define its \emph{stable manifold} to be the union of edge-triangle gradient paths that ends at $e$.  
Its \emph{unstable  manifold} is defined to be the union of vertex-edge gradient paths that begins with $e$. 
While earlier we use ``mountain ridges'' (1-stable manifolds) to motivate the graph reconstruction framework, algorithmically 
(especially for the Morse cancellations below), vertex-edge gradient paths are simpler to handle. 
Hence in our algorithm below, we in fact consider the function {\bf $\afunc = -\rho$} (instead of the density field $\rho$ itself) and the algorithm outputs (a subset of) the \emph{1-unstable manifolds} (vertex-edge paths in the discrete setting) as the recovered hidden graph. 

\subparagraph*{Morse cancellation / simplification.}
One can simplify a discrete gradient vector field $M(\K)$ (i.e, reducing the number of critical simplices) by the following Morse cancellation operation: 
A pair of critical simplices $\langle \sigma, \tau \rangle$ with $\dimension(\tau) = \dimension(\sigma)+1$ is \emph{\validcancel}, 
if there is a \emph{unique} gradient path 
$ \tau = \tau_0,\sigma_1,\ldots,\tau_\ell, \sigma_{\ell+1} = \sigma$
starting at the $k+1$-simplex $\tau$ and ends at the $k$-simplex $\sigma$. 
The \emph{Morse cancellation operation on $\langle \sigma, \tau \rangle$} then modifies the vector field $M(\K)$ by removing all gradient vectors $(\sigma_i, \tau_i)$, for $i=1, \ldots, \ell$, while adding new gradient vectors $(\sigma_i, \tau_{i-1})$, for $i = 1, \ldots, \ell+1$. Intuitively, the gradient path is inverted. Note that $\tau = \tau_0$ and $\sigma=\sigma_{\ell+1}$ are no longer critical after the cancellation as they now participate in discrete gradient vectors. 
If there is no gradient path, or more than one gradient path between this pair of critical simplices $\langle \sigma, \tau \rangle$, then this pair is \emph{not \validcancel{}} -- 
the uniqueness condition is to ensure that no cyclic V-paths are formed after the cancellation operation. 
See Figure \ref{fig:defs} (b) -- (d) for examples.  

\subsection{Persistence pairing}
\label{subsec:persist}

The Morse cancellation can be applied to any sequence of critical simplices pairs as long as they are \validcancel{} at the time of cancellation. 
There is no canonical cancellation sequence. 
To cancel features corresponding to ``noise'' w.r.t. an input piecewise-linear function $f: |\K| \to \mathbb{R}$, a popular strategy is to guide the Morse cancellation by the persistent homology induced by the lower-star filtration \cite{GDN07,2011MNRAS}, which we introduce now. 

\subparagraph*{Filtrations and lower-star filtration.}
Given a simplicial complex $\K$, let $S$ be an ordered sequence $\sigma_1, \ldots, \sigma_N$ of all $n$ simplices in $\K$ so that for any simplex $\sigma_i \in \K$, all of its faces appear before it in $S$. 
Then $S$ induces a \emph{(simplex-wise) filtration $F(\K)$:}
$ K_1 \subset K_2 \subset \cdots \subset K_N = \K, $
where $K_i = \bigcup_{j\le i} \sigma_j$ is the subcomplex formed by the prefix $\sigma_1,\ldots, \sigma_i$ of $S$. 
Passing  to homology groups, we have a \emph{persistence module} $H_*(K_1) \to \cdots \to H_*(K_N)$, which has a unique decomposition into the direct sum of a set of indecomposable summands that can be represented by the set of \emph{persistence-pairing $P(\K)$ induced by $F(\K)$}: Each \emph{persistence pair} $(\sigma_i, \sigma_j) \in P(\K)$ indicates that a new $k$-th homological class, $k = \dimension(\sigma_i)$, is created at $K_i$ and destroyed at $K_j$; $\sigma_i$ is thus called a \emph{positive simplex} as it creates, and $\sigma_j$ a \emph{negative simplex}. 
Assuming that there is a \emph{simplex-wise function} $\myf: \K \to \mathbb{R}$ such that $\myf(\sigma_i) \le \myf(\sigma_j)$ if $i < j$, 
then the \emph{persistence} of the pair $(\sigma,\tau)$ is defined as $\pers(\sigma) = \pers(\tau) = \pers(\sigma, \tau) = \myf(\tau) - \myf(\sigma)$. 
Some simplices $\sigma_\ell$'s may be unpaired, meaning that homological features created at $K_\ell$ are never destroyed. 
We augment $P(\K)$ by adding $(\sigma_\ell, \infty)$ for every unpaired simplex $\sigma_\ell$ to it, and set $\pers(\sigma_\ell, \infty) = \infty$.  

The persistent homology can be defined for any filtration of $\K$. 
In our setting, there is an input function $f: V(\K) \to \mathbb{R}$ defined at the vertices $V(\K)$ of $\K$ whose linear extension leads to a piecewise-linear (PL) function still denoted by $f: |\K| \to \mathbb{R}$.  
To reflect topological features of $f$, we use the lower-star filtration of $\K$ induced by $f$: 
Specifically, for any vertex $v\in V(\K)$, its lower-star $\LowSt(v)$ is the set of simplicies containing $v$ where $v$ has the highest $f$ value among their vertices. 
Now sort vertices of $\K$ in non-decreasing order of their $f$-values: $v_1, \ldots, v_n$.
An ordered sequence $S = \langle \sigma_1, \ldots, \sigma_N\rangle$ induces a \emph{lower-star filtration $F_f(\K)$ of $\K$ w.r.t. $f$} if $S$ can be partitioned to $n$ consecutive pieces $\langle \sigma_1, \ldots, \sigma_{I_1} \rangle$, $\langle \sigma_{I_1+1}, \ldots, \sigma_{I_2}\rangle$, $\ldots, \langle \sigma_{I_{n-1}+1}, \ldots, \sigma_{N} \rangle$, such that the $i$-th piece $\langle \sigma_{I_{i-1}+1}, \ldots, \sigma_{I_i}\rangle$ equals $\LowSt(v_i)$. 

Now let $P_f(\K)$ be the resulting set of persistence pairs induced by the lower-star filtration $F_f(\K)$. 
Extend the function $f: V(\K) \to \mathbb{R}$ to a simplex-wise function $\myf: \K \to \mathbb{R}$ where $\myf(\sigma) = \max_{v\in \sigma} f(v)$ (i.e, $\myf(\sigma)$ is the highest f-value of any of its vertices). 
For each pair $(\sigma, \tau)$, we measure its persistence by $\pers(\sigma, \tau) = \myf(\tau) - \myf(\sigma)$. 
Every simplex in $\K$ contributes to a  persistence pair in $P_f(\K)$. 
However, assuming the value of $f$ is distinct on all vertices, then those persistence pairs with zero-persistence are ``trivial'' in the sense they correspond to the local pairing of two simplices from the lower-star of the same vertex. 
A persistence pair $(\sigma, \tau)$ \emph{with positive persistence} corresponds to a pair of (homological) critical points $(p, q)$  for the PL-function $f : |\K| \to \mathbb{R}$ \cite{EH10} induced by the function $f$ on $V(\K)$, with $p\in \sigma$ and $q\in \tau$. 

\section{Reconstruction algorithm}
\label{sec:alg}

\subparagraph*{Problem setup.}
Suppose we have a domain $\Idomain$ (which will be a cube in $\mathbb{R}^d$ in this paper) and a density function $\rho: \Idomain \to \mathbb{R}$ (that ``concentrates'' around a hidden geometric graph $G \subset \Idomain$). 
In the discrete setting, our input will be a triangulation $K$ of $\Idomain$ and a density function given as a PL-function $\rho: K \to \mathbb{R}$. 
Our goal is to compute a graph $\Ghat$ approximating the hidden graph $G$. 
In Algorithm \ref{alg:WWL15}, we first present a \emph{known} discrete Morse-based graph (1-skeleton) reconstruction framework, which is based on the approaches in \cite{GDN07,DRS15,2011MNRAS,WWL15}. 

Intuitively, we wish to use ``mountain ridges'' of the density field to approximate the hidden graph, which are computed as the 1-unstable manifolds of $\afunc = -\rho$, the negation of the density function. 
Specifically, after initializing the discrete gradient vector field $\GV$ to be a trivial one, a persistence-guided Morse cancellation step is performed in Procedure \SimpVecF() to compute a new discrete gradient vector field $\GV_\delta$ so as to capture only important (high persistent) features of $\afunc$. In particular, Morse-cancellation is performed for each pair of critical simplices from $\Perall(\K)$ (if possible) in increasing order of persistence values (for pairs with equal persistence, we use the nested order as in \cite{Bauer2012}). 
Finally, the union of the 1-unstable manifolds of all remaining high-persistence critical edges is taken as the output graph $\Ghat$, as outlined in Procedure \CollectOutput().  

Since we only need 1-unstable manifolds, $\K$ is assumed to be a $2$-complex. 
It is pointed out in \cite{newpaper} that in fact, instead of performing Morse-cancellation for all critical pairs involving edges (i.e, vertex-edge pairs and edge-triangle pairs), one only needs to cancel vertex-edge pairs -- This is because only vertex-edge gradient vectors will contribute to the 1-unstable manifolds, and also new vertex-edge vectors can only be generated while canceling other vertex-edge pairs.
Hence in \SimpVecF(), we can consider \emph{only vertex-edge pairs} $(\sigma, \tau)\in P$ in order. 
Furthermore, it is not necessary to check whether the cancellation is valid or not -- it will always be valid as long as the pairs are processed in increasing orders of persistence \cite{Bauer2012}\footnote{We remark that though \cite{Bauer2012} states that the cancellation is not valid in higher dimension or non-manifold 2-complexes, all cancellations in \SimpVecF() are for vertex-edge pairs in a spanning tree which can be viewed as a 1-complex, and thus are always valid.}.

However, we can further simplify the algorithm as follows: 
First, we replace procedure \SimpVecF() by procedure \SimpVecFVtwo() as shown in Algorithm \ref{alg:WWLSimp15}, which is much simpler both conceptually and implementation speaking. Note that there is {\it no explicit cancellation operation} any more. 

\RestyleAlgo{boxruled}
\LinesNumberedHidden
\begin{algorithm}[hptb]
\caption{\WWLAlg($K$,$\rho$, $\delta$) \label{alg:WWL15}}
\DontPrintSemicolon
\KwData{Triangulation $K$ of $\Idomain$, density function $\rho: K \to \mathbb{R}$, threshold $\delta$}
\KwResult{Reconstructed graph $\Ghat$}
\Begin{
\setcounter{algoline}{0}
\Numberline Compute persistence pairings $\Perall(\K)$ by the lower-star filtration of $K$ w.r.t $\afunc = -\rho$\;
\Numberline $\GV = $\SimpVecF($\Perall(\K), \delta$)\;
\Numberline$\Ghat$ = \CollectOutput($\GV$)\;
\Numberline \Return $\Ghat$
}
\SetKwProg{myproc}{Procedure}{}{}
\myproc{\SimpVecF($\Perall(\K), \delta$)}{
\setcounter{algoline}{0}
\Numberline Set initial discrete gradient field $\GV$ on $K$ to be trivial\;
\Numberline Rank all persistence pairs in $\Perall(\K)$ in increasing order of their persistence\;
\Numberline\For{each $(\sigma, \tau) \in \Perall(\K)$ with $\pers(\sigma, \tau)\le \delta$}{
\Numberline If possible, perform discrete-Morse cancellation of $(\sigma,\tau)$ and update the discrete gradient vector field $\GV$ \;}
\Numberline \Return $\GV$
}
\SetKwProg{myproc}{Procedure}{}{}
\myproc{\CollectOutput($\GV$)}{
\setcounter{algoline}{0}
\Numberline $\Ghat=\emptyset$\;
\Numberline\For{each remaining critical edge $e$ with $\pers(e)> \delta$}{
\Numberline$\Ghat = \Ghat\bigcup\{$1-unstable manifold of $e\}$\;}
\Numberline\Return $\Ghat$
}
\end{algorithm}

\RestyleAlgo{boxruled}
\LinesNumberedHidden
\begin{algorithm}[hptb]
\caption{\WWLAlgSimp($K$,$\rho$, $\delta$) \label{alg:WWLSimp15}}
\DontPrintSemicolon
\SetKwProg{myproc}{Procedure}{}{}
\myproc{\SimpVecFVtwo($\Perall(\K), \delta$) $/*$ This procedure replaces original \SimpVecF() $*/$ \label{alg:FVtwo}}{
\setcounter{algoline}{0}
\Numberline$\Pi :=$ the set of vertex-edge persistence pairs from $ \Perall(\K)$\;
\Numberline Set $\Pi_{\le \delta} \subseteq \Pi$ to be $\Pi_{\le \delta} = \{(v,e)\in \Pi \mid \pers(v,e) \le \delta\}$\;
\Numberline$\mytree := \bigcup_{(v, \sigma) \in \Pi_{\le \delta}} \{ \sigma=\langle u_1, u_2\rangle, u_1, u_2 \}$\;
\Numberline\Return $\mytree$\;

}
\SetKwProg{myproc}{Procedure}{}{}
\myproc{\CollectOutputVtwo($\mytree$) $/*$ This procedure replaces \CollectOutput() $*/$\label{alg:OutputVtwo}}{
\setcounter{algoline}{0}
\Numberline$\Ghat=\emptyset$\;
\Numberline\For{each edge $e=\langle u, v \rangle$ with $\pers(e) > \delta$}{
 \Numberline Let $\pi(u)$ be the unique path from $u$ to the sink of the tree $T_i$ containing $u$\; 
 \Numberline Define $\pi(v)$ similarly; ~~Set $\Ghat = \Ghat \cup \pi(u) \cup \pi(v) \cup \{e\}$ 
 }
\Numberline\Return $\Ghat$
}
\end{algorithm}

\vspace*{-0.15in}
The $1$-dimensional simplicial complex $\mytree$ output by procedure \SimpVecFVtwo() may have multiple connected components $\mytree = \{T_1, \ldots, T_k\}$ -- In fact, it is known that each $T_i$ is a tree and $\mytree$ is a forest (see results from 
\cite{attali2009persistence,Bauer2012} as summarized in Lemma \ref{lem:01cripairs} below). 
For each component $T_i$, we define its \emph{sink}, denoted by $\mysink(T_i)$, as the vertex $\myv_i \in T_i$ with the lowest function $\afunc = -\rho$ value. 
Lemma \ref{lem:01cripairs} also states that the sink of $T_i$ would have been the \emph{only critical simplex} among all simplices in $T_i$, if we had performed the $\delta$-simplification as specified in procedure \SimpVecF(). 
Next, we replace procedure \CollectOutput() by procedure \CollectOutputVtwo() shown in Algorithm \ref{alg:WWLSimp15}. 
We use \WWLAlgSimp() to denote our simplified version of Algorithm \ref{alg:WWL15} (with \SimpVecF() replaced by \SimpVecFVtwo(), and \CollectOutput() replaced by \CollectOutputVtwo(). 
In summary,
algorithm \WWLAlgSimp($K, \rho, \delta$) works by first computing all persistence pairs as before. It then collects all vertex-edge persistence pairs $(v,e)$ with $\pers(v,e)\leq\delta$. These edges along with the set of all vertices form a spanning forest $\mytree$. 
Then, for every edge $e=\langle u,v\rangle$ with $\pers(e)>\delta$, it outputs the 1-unstable manifold of $e$, which is simply the union of $e$ and the unique paths from $u$ and $v$ to the sink (root) of the tree containing them respectively.
Its time complexity is stated below; note for the previous algorithm \WWLAlg(), the cancellation step can take $\tilde{O}(n^2)$ time. 
\begin{theorem}
The time complexity of our Algorithm \SimpVecF() is $O(Pert(\K)+n)$,
where $PerT(\K)$ is the time to compute persistence pairings for $\K$, and $n$ is the total number of vertices and edges in $\K$. 
\end{theorem}

We remark that the $O(n)$ term is contributed by the step collecting all $1$-unstable manifolds, which takes linear time if one avoids revisiting edges while tracing the paths.
\subparagraph*{Justification of the modified algorithm \WWLAlgSimp().}
Let $\GV_\delta$ denote the resulting discrete gradient field after canceling all \emph{vertex-edge} persistence pairs in $\Perall(\K)$ with persistence at most $\delta$; that is, $\GV_\delta$ is the output of the procedure \SimpVecF() (although we only compute the relevant part of the discrete gradient vector field). 
Using observations from \cite{attali2009persistence,Bauer2012},
we show that the output $\mytree$ of procedure \SimpVecFVtwo() includes all information of $\GV_\delta$. Furthermore, procedure \CollectOutputVtwo() computes the correct 1-unstable manifolds for all critical edges with persistence larger than $\delta$. 
Indeed, observe that edges in Morse pairings from $\GV_\delta$ (for any $\delta \ge 0$) form a spanning forest of edges in $\K$. 
Results of \cite{Bauer2012} imply that the output $\mytree$ constructed by our modified procedure corresponds exactly to this spanning forest:  

\begin{lemma}\label{lem:01cripairs}
The following statements hold for the output $\mytree$ of procedure \SimpVecFVtwo()
w.r.t any $\delta \ge 0$: 
\begin{itemize}
\item[(i)] $\mytree$ is a spanning forest consisting of potentially multiple trees $\{T_1, \ldots, T_k\}$. 
\item[(ii)] For each tree $T_i$, its sink $\myv_i$ is the only critical simplex in $\GV_\delta$. The collection of $\myv_i$s corresponds exactly to those vertices whose persistence is bigger than $\delta$. 
\item[(iii)] Any edge with $\pers(e) > \delta$ remains critical in $\GV_\delta$ (and cannot be contained in $\mytree$). 
\end{itemize}
\end{lemma}

\begin{wrapfigure}{r}{0.12\textheight}
\centering
\vspace*{-0.2in}\includegraphics[width=.12\textwidth]{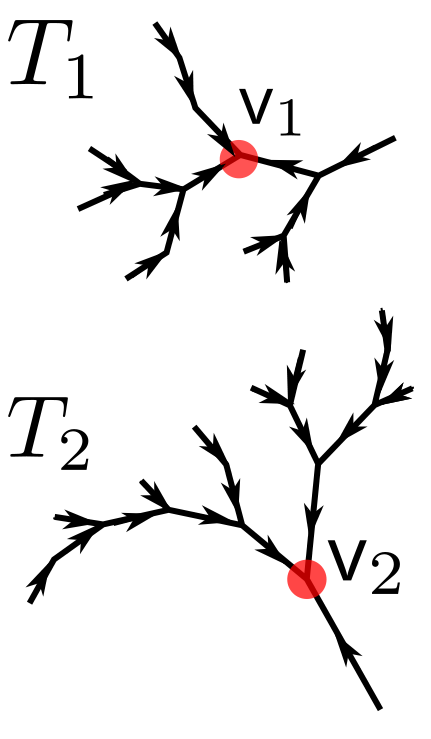}
\label{fig:tree}
\end{wrapfigure}
Note that, (ii) above implies that for each $T_i$, any discrete gradient path of $\GV_\delta$ in $T_i$ terminates at its sink $\myv_i$. See the right figure for an illustration. Hence for any vertex $v\in T_i$, the path $\pi(v)$ computed in procedure \CollectOutputVtwo() is the unique discrete gradient path starting at $v$. 
This immediately leads to the following result: 

\begin{corollary}\label{cor:output}
For each critical edge $e = \langle u, v \rangle$ with $\pers(e) \ge \delta$, $\pi(u) \cup \pi(v)\cup \{e\}$ as computed in procedure \CollectOutputVtwo() is the 1-unstable manifold of $e$ in $\GV_\delta$. 
Hence the output of our simplified algorithm \WWLAlgSimp() equals that of the original algorithm \WWLAlg(). 
\end{corollary}

\section{Noise model}
\label{sec:noisemodel}
We first describe the noise model in the continuous setting where the domain is $\Idomain = [0,1]^d$. We then explain the setup in the discrete setting when the input is a triangulation $\K$ of $\Idomain$. 

\begin{wrapfigure}{r}{0.18\textheight}
\centering
\vspace*{-0.1in}\includegraphics[width=.2\textwidth]{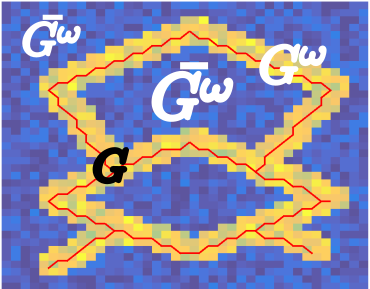}
\label{fig:noise-model}
\end{wrapfigure}
Given a connected ``true graph'' $\trueG \subset \Idomain$, consider a \emph{$\myw$-neighborhood} $\thickenG \subseteq \Idomain$, meaning that (i) $\trueG \subseteq \thickenG$, and (ii) for any $x\in \thickenG$, $d(x, \trueG) \le \myw$ 
(i.e, $\thickenG$ is sandwiched between $\trueG$ and its $\myw$-offset). 
Given $\thickenG$, we use $\redR$ to denote the closure of its complement $\redR = \mathrm{cl}(\Idomain \setminus \thickenG)$.
See the right figure, showing $G$ (red graph) with its $\myw$-neighborhood $\thickenG$ (orange). 

\begin{definition}
A density function $\rho: \Idomain \to \mathbb{R}$ is a \emph{$(\beta, \nu, \myw)$-approximation of a connected graph $\trueG$} if the following holds: 
\begin{itemize}
\item[C-1] There is a $\myw$-neighborhood  $\thickenG$ of $\trueG$ such that $\thickenG$ deformation retracts to $\trueG$. 
\item[C-2] $\rho(x) \in [\beta,\beta+\nu]$ for $x\in \thickenG$; and $\rho(x) \in [0, \nu]$ otherwise. 
Furthermore, $\beta > 2\nu$. 
\end{itemize}
\end{definition}

Intuitively, this noise model requires that the density $\rho$ concentrates around the true graph $\trueG$ in the sense that the density is significantly higher inside $\thickenG$ than outside it; and  
the density fluctuation inside or outside $\thickenG$ is small compared to the density value in $\thickenG$ (condition C-2). 
Condition C-1 says that the neighborhood has the same topology of the hidden graph. 
Such a density field could for example be generated as follows: 
Imagine that there is an ideal density field $f_G: \Idomain \to \mathbb{R}$ where $f_G(x) = \beta$ for $x\in \thickenG$ and $0$ otherwise. 
There is a noisy perturbation $g: \Idomain \to \mathbb{R}$ whose size is always bounded by $g(x) \in [0, \nu]$ for any $x\in \Idomain$. 
The observed density field $\rho = f_G + g$ is an $(\beta, \nu, \myw)$-approximation of $\trueG$. 

In the discrete setting when we have a triangulation $\K$ of $\Idomain$, we define a $\myw$-neighborhood $\thickenG$ to be a subcomplex of $\K$, i.e, $\thickenG \subseteq \K$, such that (i) $\trueG$ is contained in the underlying space of $\thickenG$ and (ii) for any vertex $v \in V(\thickenG)$, $d(v, \trueG) \le \myw$. 
The \redregion{} $\redR \subseteq \K$ is simply the smallest subcomplex of $\K$ that contains all simplices from $\K \setminus \thickenG$ (i.e, all simplices {\bf not} in $\thickenG$ and their faces). A PL-function $\rho: \K \to \mathbb{R}$ $(\beta, \nu, \myw)$-approximation of $\trueG$ can be extended to this setting by requiring the underlying space of $\thickenG$ deformation retracts to $\trueG$ as in (C-1), and having those density conditions in (C-2) only at vertices of $\K$. 

We remark that the noise model is still limited -- In particular, it does not allow significant non-uniform density distribution. However, this is the first time that theoretical guarantees are provided for a discrete Morse based reconstruction framework, despite that such a framework has been used for different applications before. We also give experiments and discussions in Appendix \ref{appendix:sec:exp} that the algorithm works beyond this noise model empirical, where thresholding type approaches do not work. 

\section{Theoretical guarantee}
\label{sec:general}
In this section, we prove results that are applicable to any dimension. 
Recall that $\GV_\delta$ is the discrete gradient field after the \emph{$\delta$-Morse cancellation process}, where we perform Morse-cancellation for all \emph{vertex-edge} persistence pairs from $\Perall(K)$. (While our algorithm does not maintain $\GV_\delta$ explicitly, we use it for theoretical analysis.) 
At this point, all positive edges (i.e, those paired with triangles or unpaired in $\Perall(K)$) remain critical in $\GV_\delta$.
Some negative edges (i.e, those paired with vertices in $\Perall(K)$) are also critical in $\GV_\delta$ -- these are exactly the negative edges with persistence bigger than $\delta$. 
\CollectOutputVtwo() only takes the 1-unstable manifolds of those critical edges (positive or negative) with persistence bigger than $\delta$; so those positive edges whose persistence is $\le \delta$ (if there is any) are ignored. 

From now on, we use ``\emph{under our noise model}'' to refer to (1)  the input is a ($\beta, \nu, \myw$)-approximated density field w.r.t. $\trueG$, and (2) $\delta \in [\nu, \beta-\nu)$. Let $\Ghat$ be the output of algorithm \WWLAlgSimp($\K, \rho, \delta$).
The proof of the following result is in 
Appendix \ref{appendix:prop:crit-edge-vertex}. 

\begin{proposition} Under our noise model, we have: 
\begin{itemize}
\item[(i)] There is a single critical vertex left after \SimpVecF() which is in $\thickenG$.
\item[(ii)] Every critical edge considered by \CollectOutputVtwo() {forms a persistence pair}  with a triangle. 
\item[(iii)] Every critical edge considered by \CollectOutputVtwo() is in $\thickenG$. 
\end{itemize}
\label{prop:crit-edge-vertex}
\end{proposition}

\begin{theorem}
Under our noise model, the output graph satisfies $\hat{G} \subseteq \thickenG$.
\label{thm:contain}
\end{theorem}
\begin{proof}
Recall that the output graph $\Ghat$ consists of the union of 1-unstable manifolds of all the edges $e^*_1, \ldots, e^*_g$ with persistence larger than $\delta$ -- By Propositions~\ref{prop:crit-edge-vertex} (ii) and (iii), they are all positive (paired with triangles), and contained inside $\thickenG$. 

Take any $i \in [1, g]$ and consider $e^*_i = \langle u, v \rangle$. Without loss of generality, consider the gradient path starting from $u$: 
$\pi: u = u_1, e_1, u_2, e_2, \ldots, u_s, e_s, u_{s+1}. $
By Lemma \ref{lem:01cripairs} and Proposition \ref{prop:crit-edge-vertex}, $u_{s+1}$ must be a critical vertex (a sink) and is necessarily the global minimum $v_0$, which is also contained inside $\thickenG$. 
We now argue that the entire path $\pi$ (i.e, all simplices in it) is contained inside $\thickenG$. 
In fact, we argue a stronger statement: First, we say that a gradient vector $(v, e)$  is \emph{crossing} if $v\in \thickenG$ and $e \notin \thickenG$ (i.e, $e\in \redR$) -- Since $v$ is an endpoint of $e$, this means that the other endpoint of $e$ must lie in $\K \setminus \thickenG$. 

\begin{claim}\label{lemma:nocrossing}
During the $\delta$-Morse cancellation, no crossing gradient vector is ever produced. 
\end{claim}
\begin{proof}
Suppose the lemma is not true: Then let $(v, e)$ be the \emph{first} crossing gradient vector ever produced during the $\delta$-Morse cancellation process. 
Since we start with a trivial discrete gradient vector field, the creation of $(v,e)$ can only be caused by reversing of some gradient path $\pi'$ connecting two critical simplices $v'$ and $e'$ while we are performing Morse-cancellation for the persistence pair $(v', e')$. 
Obviously, $\pers(v', e') \le\delta$. 
On the other hand, due to our $(\beta,\nu,\myw)$-noise model and the choice of $\delta$, it must be that either both $v', e' \in \thickenG$ or both $v', e' \in \K \setminus \thickenG$ -- as otherwise, the persistence of this pair will be larger than $\beta - \nu > \delta$. 

Now consider this gradient path $\pi'$ connecting $v'$ and $e'$ in the current discrete gradient vector field $\GV'$. Since the pair $(v, e)$ becomes a gradient vector after the inversion of this path, it must be that $(w,e)$
currently is a gradient vector where $e = \langle v, w \rangle$. Furthermore, since the path $\pi'$ begins and ends with simplices either both in $\thickenG$ or both outside it,  the path $\pi'$ must contain a gradient vector $(v'', e'')$ going in the opposite direction crossing inside/outside, that is, $v'' \in \thickenG$ and $e'' \notin \thickenG$. In other words, it must contain a crossing gradient vector. This however contradicts to our assumption that $(v,e)$ would be the first crossing gradient vector. Hence the assumption is wrong and no crossing gradient vector can ever be created. 
\end{proof}

\noindent As there is no crossing gradient vector during and after $\delta$-Morse cancellation, it follows that $\pi$, which is one piece of the 1-unstable manifold of the critical edge $e^*_i$, has to be contained inside $\thickenG$. The same argument works for the other piece of $1$-unstable manifold of $e^*_i$ (starting from the other endpoint of $e^*_i$). Since this is for any $i\in [1,g]$, the theorem holds. 
\end{proof}

The previous theorem shows that $\Ghat$ is close to $G$ in geometry. Next we will show that they are also close in topology. 

\begin{proposition}
Under our noise model, $\Ghat$ is homotopy equivalent to $G$.
\label{prop:same-genus}
\end{proposition}
\begin{proof}
We show that $\Ghat$ has the same first Betti number as that of $G$ which implies the claim as any two graphs in $\mathbb{R}^d$ with the same first Betti number are homotopy equivalent.

The underlying space of $\omega$-neighborhood $\thickenG$ of $G$ deformation retracts to $G$ by definition. Observe that, by our noise model, $\thickenG$ is a sublevel set in the filtration that determines the persistence pairs. This sublevel set being homotopy equivalent to $G$ must contain exactly $g$ positive edges where $g$ is the first Betti number of $G$. Each of these positive edges pairs with a triangle in $\overline{\thickenG}$. Therefore, $\pers(e,t)>\delta$ for each of the $g$ positive edges in $\thickenG$. By our earlier results, these are exactly the edges that will be considered by procedure \CollectOutputVtwo(). Our algorithm constructs $\Ghat$ by adding these $g$ positive edges to the spanning tree each of which adds a new cycle. Thus, $\Ghat$ has first Betti number $g$.
\end{proof}

We have already proved that $\Ghat$ is contained in $\thickenG$. This fact along with Proposition~\ref{prop:same-genus} can be used to argue that any deformation retraction taking (underlying space) $\thickenG$ to $G$ also takes $\Ghat$ to a subset $G'\subseteq G$ where $G'$ and $G$ have the same first Betti number. In what follows, we use $\thickenG$ to denote also its underlying space.

\begin{theorem}
Let $F: \thickenG \times [0,1]\rightarrow \thickenG$ be any deformation retraction. Then, the restriction $F|_{\Ghat}: \Ghat \times [0,1]\rightarrow \thickenG$  is a homotopy from the embedding $\Ghat$ to $G'\subseteq G$ where $G'$ is the minimal subset so that $G$ and $G'$ have the same first Betti number.
\end{theorem}
\begin{proof}
The fact that $F|_{\Ghat}(\cdot,\ell)$ is continuous for any $\ell\in[0,1]$ is obvious from the continuity of $F$. Only thing that needs to be shown is that $F|_{\Ghat}(\Ghat,1)=G'$. Suppose not. Then, $G''=F|_{\Ghat}(\Ghat,1)$ is a proper subset of $G$ which has a first Betti number less than that of $G$.

\noindent We observe that the cycle in $\Ghat$ created by a positive edge $e$ along with the paths to the root of the spanning tree is also non-trivial in $\thickenG$ because this is a cycle created by adding the edge $e$ during persistence filtration and the edge $e$ is not killed in $\thickenG$.Therefore, a cycle basis for $\Ghat$ is also a homology basis for $\thickenG$. 
Since the map $F(\cdot,1): \thickenG\rightarrow G$ is a homotopy equivalence, it induces an isomorphism in the respective homology groups; in particular, a homology basis in $\thickenG$ is mapped to a homology basis in $G$. Therefore, the image $G''=F|_{\Ghat}(\Ghat,1)$ must have a basis of cardinality $g$ if $\Ghat$ has first Betti number $g$. But, $G''$ cannot have a cycle basis of cardinality $g$ if it is a proper subset of $G'$ reaching a contradiction.
\end{proof}

\section{Additional guarantee for 2D}
\label{sec:2D}
For $\mathbb{R}^2$, we now show that $\thickenG$ actually deformation retracts to $\Ghat$, which is stronger than saying $G$ and $\Ghat$ are homotopy equivalent.
We are unable to prove this result for dimensions higher than 2, as our current proof needs that the edge-triangle persistence pairs can always be canceled (even though our algorithm does not depend on edge-triangle cancellations at all). It would be interesting, as a future work, to see whether a different approach can be developed to avoid this obstruction for the special case under our noise model.
The main result of this section is as follows.
\begin{theorem}
Under our noise model, $\thickenG$ deformation retracts to $G$ and $\hat{G}$.
\label{thm:deform-retract}
\end{theorem}
This main result follows from Proposition~\ref{deform-lem} and Theorem~\ref{deform-thm} below. To prove them, we will show that there exists a partition ${\mathcal R}:=\{R_i\}$ of the set of triangles in $\K$ for which Theorem~\ref{deform-thm} holds. (This theorem is our main tool in establishing the deformation retract.) 
We first state the results below before giving their proofs. Let $B_i=\partial R_i$ where $\partial$ is the boundary operator operating on the $2$-chain $R_i$. We also abuse the notations $R_i$ and $B_i$ to denote the geometric space that is the point-wise union of simplices in the respective chains. Let $t_i$ be a triangle in $R_i$ whose choice will be explained later.
In the following, let $H$ be the maximal set of edges in $\Ghat$ whose deletions do not eliminate a cycle (assume that a vertex is deleted only if all of its edges are deleted). Observe that $H$ necessarily consists of negative edges forming ``hairs'' attached to the loops of $\Ghat$ and hence to $\cup_i B_i$ because of the following proposition.

\begin{proposition}
Under our noise model, $\Ghat=\cup B_i \bigcup H$.
 \label{deform-lem}
\end{proposition}

\begin{theorem}
Under our noise model, there exists a partition $\{ R_i\}$ of triangles in $K$ such that, there is a deformation retraction of $\cup_i (R_i\setminus t_i)$ to $\Ghat$ that comprises of two deformation retractions, one from $\cup_i (R_i\setminus t_i)$ to $\thickenG$ and another one from $\thickenG$ to $\cup_i B_i\bigcup H$ which is $\Ghat$.
\label{deform-thm}
\end{theorem}

Now we describe the construction of a partition ${\mathcal R}$ of the triangles in $\K$ to prove Proposition \ref{deform-lem} and Theorem \ref{deform-thm}. 
For technicality we assume that $\K$ is augmented to a triangulation of a sphere by putting a vertex $v$ at infinity and joining it to the boundary of $\K$ with edges and triangles all of whom have function value $\infty$. 
Let $P(\K)$ be the collection of persistence pairs of the form either $(\sigma, \tau)$ or $(\sigma, \infty)$ generated from the lower-star filtration $F(\K)$ as described before. Since $\K$ is 2-dimensional, each pair $(\sigma, \tau)$ is either a vertex-edge pair or an edge-triangle pair. 
We order persistence pairs in $P(\K)$
by their persistence, where ties are broken via the nested order in the filtration $F(\K)$, and obtain: 
\begin{align}\label{eqn:orderperpairs}
P(\K) &= \{ (\sigma_1, \tau_1), \ldots, (\sigma_n, \tau_n), (\alpha_1, \infty), \ldots, (\alpha_s, \infty)\}.
\end{align}

Starting with a trivial discrete gradient vector field $\GV_0$ where all simplices in $\K$ are critical, the algorithm \SimpVecF() performs Morse cancellations for the first $m\le n$ persistence pairs $(\sigma_1, \tau_1),\ldots,(\sigma_m, \tau_m)$ in order where $\pers(\sigma_m,\tau_m)\leq \delta$ but $(\sigma_{m+1}, \tau_{m+1}) > \delta$, Let $M_i$ denote the gradient vector field after canceling $(\sigma_i,\tau_i)$. Recall that in the implementation of the algorithm we do not need to perform Morse cancellation for any edge-triangle pairs. 
However in this section, for the theoretical analysis, we will cancel edge-triangle pairs as well. 
Recall that a positive edge is one that creates a 1-cycle, namely, it is either paired with a triangle or unpaired; while a negative edge is one that destroys a 0-cycle (i.e, paired with a vertex). 

Consider the ordered
sequence of edge-triangle persistence pairs, $(e_1,t_1),\ldots,(e_n,t_n)$, which is a subsequence of the one in (\ref{eqn:orderperpairs}). 
Consider the sequence $t_1,\dots,t_n$ of triangles in $\K$ 
ordered by the above sequence.
Recall the standard persistence algorithm \cite{ELZ02}. It implicitly associates a $2$-chain with a triangle $t$ when searching for the edge it is about to pair with. This $2$-chain is non-empty if $t$ is a destructor, and is empty otherwise. Let $D_i$ denote this $2$-chain associated with $t_i$ for $i\in[1,n]$.
Initially, the algorithm asserts $D_i=t_i$. At any stage, if $D_i$ is
not empty, the persistence algorithm identifies the edge $e$ in the boundary $\partial D_i$ that has been inserted into the filtration $F(\K)$ most recently. If $e$ has not been paired with anyone,
the algorithm creates the persistence pair $(e,t_i)$. Otherwise, if $e$ has already been paired with a triangle, say $t_{i'}$, then $D_i$ is updated with $D_i=D_i+ D_{i'}$ and the search continues.
Given an index $j\in[1,n]$, we define a modified set of chains ${C}_i^j$ inductively as follows. For $j=1$, ${C}_i^1=t_i$. Assume that $C_i^{j-1}$ has been
already defined. To define $C_i^j$, similar to the persistence algorithm, check if the edge $e_{j-1}$
% \jiayuan{I think ``edge $e_{j-1}$ in the persistence order'' already defines  $e_{j-1}$}
is on the boundary $\partial C_i^{j-1}$. If so, define $C_{i}^j:=C_i^{j-1}+C_{j-1}^{j-1}$ 
and
$C_i^j:=C_i^{j-1}$ otherwise.
The following result is proved in Appendix \ref{appendix:modchain-prop}. 
\begin{proposition}
	For $i\in[1,n]$, 
    $e_i$ is in $\partial {C}_i^{i}$. Furthermore, $e_i$ is
	the most recent edge in $\partial {C}_i^{i}$ according to the filtration order $F(\K)$.
	\label{modchain-prop}
\end{proposition}

Procedure \SimpVecF() also implicitly maintains a $2$-chain $R_i^*$ with each triangle $t_i$.
Initially, $R_i^*=t_i$ as in the case of ${D}_i$. Then, inductively assume that $R_i^*$ is the $2$-chain implicitly associated with $t_i$ when a persistence pair $(e_{i'},t_{i'})$ is about to be considered by 
\SimpVecF() and the boundary $\partial R_i^*$ contains $e_{i'}$. By reversing a gradient path between $t_{i'}$ and $e_{i'}$, it implicitly updates
the $2$-chain $R_i^*$ as $R_i^*:=R_i^*+R_{i'}^*$. We observe that $R_i^*$ is identical with $C_i^{i'}$.
Proposition~\ref{main-prop} below establishes this fact along with some other inductive properties useful to prove Theorem~\ref{deform-thm}. 
The proof can be found in Appendix \ref{appendix:main-prop}.
\begin{proposition}
Let $(e_j,t_j)$ be the edge-triangle persistence pair \SimpVecF$()$ is about to consider and
let $C_i^{j}$ be the $2$-chains defined as above . Then, the following statements hold:
\noindent
\begin{enumerate}[label=(\alph*)]
\item For each triangle $t_i$, $i=1,\ldots,n$, in the persistence order, the $2$-chain $R_i^*$ satisfies the following conditions:(a.i) $R_i^*=C_i^j$, (a.ii) interpreting $R_i^*$ as a set of triangles, one has that the sets $R_i^*$, $i=j,\ldots, n$, partition the set of all triangles in $\K$.
\item There is a gradient path from $t_i$ to all edges of the triangles in $R_i^*$, and (b.i) the path is unique if the edge is in the boundary $\partial R_i^*$ for every $i=j,\ldots, n$; (b.ii) if there is more than one gradient path from $t_i$ to an edge $e$, then $e$ must be a negative edge. 
\end{enumerate}
\label{main-prop}
\end{proposition}

We are now ready to setup the regions $R_i$s needed for Theorem \ref{deform-thm} and Proposition \ref{deform-lem}. Suppose the first $m$ edge-triangle pairs have persistence less than or equal to $\delta$, the parameter supplied to \SimpVecF(). Then, we set $R_i=R_i^*$ as in Proposition~\ref{main-prop} for $i\geq  m+1,\ldots,n$. 
The proof for Proposition \ref{deform-lem} is in Appendix \ref{appendix:deform-lem}. 

Finally, similar to the vertex-edge gradient vectors, we say that a gradient vector $(e, t)$  is \emph{crossing} if $e\in \thickenG$ and $t \notin \thickenG$. 
The following claim can be proved similarly as Claim~\ref{lemma:nocrossing}.

\begin{claim} \label {crossing}
During the $\delta$-Morse cancellation of edge-triangle pairs, no crossing gradient vector is ever produced. 
\end{claim}

{\noindent {\bf Proof of Theorem \ref{deform-thm}.}}
Set $\hat{R}_i=R_i\setminus t_i$. 
Let $\mytree$ be the spanning tree formed by all negative edges and their vertices. Let $L_i$ be the set of edges in $R_i$ that has more than one gradient path from $t_i$ to them; $L_i \subset \mytree$ by Proposition \ref{main-prop} (b.ii). 
First, we want to establish a deformation retraction from $\cup (\hat{R_i}\setminus {t_i})$ to $\thickenG$. 
To do this, for $k= 0,1,\ldots,s$, we will define $\hat{R}_i^k$ 
inductively where $\hat{R}_i^{k-1}$ deformation retracts to $\hat{R}_i^k$ and $\hat{R}_i^s \subseteq \thickenG\cup L_i$. 
Let $\hat{R}^0_i=\hat{R}_i$. For $k=1,\ldots,s$, consider a \emph{positive} edge $e$ in $\hat{R}_i^{k-1}$ where (a) $e$ is not in $\thickenG$ and (b) there is a unique gradient path in $R_i$ from $t_i$ to $e$ that passes through triangles all of which are in $R_i\setminus \hat{R}_i^{k-1}$. If such an edge $e$ exists, then $e$ is necessarily incident to a single triangle, say $t$, in $\hat{R}_i^{k-1}$. We collapse the pair $(e,t)$, which is necessarily an edge-triangle gradient vector pair because $e$ is positive. 
We take $\hat{R}_i^k$ to be $\hat{R}_i^{k-1}\setminus \{e,t\}$. 
If no such $e$ exists, then either (A) there is no positive edge in $\hat{R}_i^{k-1} \setminus \thickenG$ any more; or (B) for each positive edge $e' \in \hat{R}_i^{k-1} \setminus \thickenG$, (B-1) there is a unique gradient path from $t_i$ to $e'$ but this path passes through some triangle in $\hat{R}_i^{k-1}$; or (B-2) there are two gradient paths from $t_i$ to $e'$.

If there is no positive edge in $\hat{R}_i^{k-1} \setminus \thickenG$ any more, then $\hat{R}_i^{k-1} \subseteq \thickenG \cup L_i$, as otherwise, there will be at least some triangle from $\hat{R}_i^{k-1} \setminus \thickenG \cup L_i$ with at least one boundary edge of it being positive. The induction then terminates; we set $s=k-1$ and reach our goal.

We now show that case (B-1) is not possible. 
Suppose it happens, that is, $e'\in \hat{R}_i^{k-1} \setminus L_i$ is an edge not in $\thickenG$ for which the unique gradient path from $t_i$ passes through triangles in $\hat{R}_i^{k-1}$. Let $e''$ be the first edge in this path that is in $\hat{R_i}^{k-1} \setminus L_i$. Then, if $e''\not\in \thickenG$, it qualifies for the conditions (a) and (b) required for $e$ reaching a contradiction. So, assume $e''\in \thickenG$. But, in that case, we have a gradient path that goes into $\thickenG$ and then comes out to reach $e'\not\in \thickenG$. There has to be a gradient pair in this path where the edge is in $\thickenG$ and the triangle is not in $\thickenG$. This contradicts Claim~\ref{crossing}. Thus, case (B-1) is not possible. 
Now consider (B-2):
$e'$ must be negative by Proposition \ref{main-prop} (b.ii). So, it is not possible either.

To summarize, the induction terminates in case (A), at which time we would have that $\hat{R}_i^s \subseteq \thickenG \cup L_i$. 
Furthermore, this process also establishes a deformation retraction from $\hat{R}_i$ to 
$\hat{R}_i^s$ realized by successive collapses of edge-triangle pairs.
Furthermore, by construction, each collapsed pair $(e,t)$ must be from $\redR{}$, hence $\cup_i \hat{R}_i^s$ contains all simplices in $\thickenG$.
Combined with that $\hat{R}_i^s \subseteq \thickenG \cup L_i$, we have that $\cup_i \hat{R}_i^s = \thickenG \cup L$, where $L = \cup_i L_i$ is a subset of the spanning tree $\mytree$. 
The edges in $L$ being part of a spanning tree cannot form a cycle and thus can be retracted along the tree to $\thickenG$, which gives rise to
a deformation retraction from $\cup_i(R_i \setminus t_i)$ to $\thickenG \cup L$ and then to $\thickenG$, establishing the first part of Theorem \ref{deform-thm}. 

We now show that
$(\cup_i \hat{R}_i^s) \bigcap \thickenG = \thickenG$ deformation retracts to $\cup B_i \bigcup H$. 
Let $\hat{L}_i$ be the edges in $\hat{R}_i^s \cap \thickenG$ with more than one gradient path from $t_i$ to them. These edges are negative by Proposition \ref{main-prop} (b.ii). 
Replacing $\hat{R}_i$ with $\hat{R}_i^s \cap \thickenG$ and edges in $B_i\cup \hat{L}_i$ playing the role of edges in $\thickenG \cup L_i$ in the above induction, we can obtain that $\hat{R}_i^s \cap \thickenG$ deformation retracts to $B_i\cup \hat{L}_i$. Observe that now instead of Claim 2, we use the fact that no edge-triangle gradient path crosses $B_i$ that consists of only negative and critical edges. 
To this end, we also observe that $\cup \hat{L}_i = \mathcal{T}\cap \thickenG$ where $\mathcal{T}$ is the spanning tree formed by all negative edges, as we only collapse edge-triangle pairs that are gradient pairs (hence the participating edges are always positive). This implies that $H \subset \cup \hat{L}_i$. 
Again, edges in $\hat{L}_i$ (being part of a spanning tree) can be retracted along the spanning tree till one reaches $B_i$ or edges in $H$. Performing this for each $i$, we thus obtain a deformation retraction from $(\cup_i \hat{R}_i^s) \bigcap \thickenG = \thickenG$ to $\cup B_i\bigcup \cup \hat{L}_i$and further to $\cup B_i \bigcup H = \Ghat$.
This finishes the proof of Theorem \ref{deform-thm}.

In Appendix \ref{appendix:sec:exp}, we also provide some experiments demonstrating the efficiency of the simplified algorithm, as well as discussion on thresholding strategies. 

\subparagraph*{Acknowledgments}
We thank Suyi Wang for generously helping with the software. The Enzo dataset used in our experiments is obtained from \cite{ENZO_code}. The Bone dataset is obtained from \cite{ct_bone}. This work is supported by National Science Foundation under grants CCF-1526513, CCF-1618247, and CCF-1740761, and by National Institute of Health under grant R01EB022899.
\bibliography{discreteMS}

\newpage
\appendix
\section{Missing proofs}

\subsection{Proof of Proposition \ref{prop:crit-edge-vertex}}
\label{appendix:prop:crit-edge-vertex}

\subparagraph*{Proof of claim (i):}
We first show that there can be no critical vertex of $\GV_\delta$ from the interior of \redregion{} $\redR$. 
Indeed, suppose there is a critical vertex $v \notin \thickenG$, then it forms a persistence pair  either with $\infty$ or with a critical edge $e\in K$ in $\Perall(K)$. 
The former cannot happen as the only vertex unpaired by the persistence algorithm is the global minimum of the input function $f$, which necessarily lies in the \greenregion{}  $\thickenG$. 
So suppose $(v, e)$ is the persistent pair containing $v$. 
It then follows that $e \in \redR$ as it must come after $v$ in the lower-star filtration $F(\K)$ induced by $f$. 
Hence $\pers(v,e)$ is necessarily smaller than $\delta$ under our noise model. 
In other words, the persistence pairing $(v,e)$ should have already been canceled during the $\delta$-Morse simplification process. 
As a result, there cannot be any critical vertex left in $\redR$. 

Next, we argue that there is exactly one critical vertex in the \greenregion{} $\thickenG$ in the final discrete gradient vector field $\GV_\delta$. 
First, note that the global minimum $v_0$ of the PL-function $f$ must remain critical, as $v_0$ will be paired with $\infty$ by the persistent homology induced by the lower-star filtration, and the persistence pairing $(v_0, \infty)$ remains after the $\delta$-Morse cancellation by Lemma \ref{lem:01cripairs} (ii). 
We now prove that there cannot be any other critical vertex in $\thickenG$.

Assume on the contrary that there is another critical vertex $u\in \thickenG$. This means we have a persistence pairing $(u, e')$ with $\pers(u, e') > \delta$. It then follows that $e' \in \redR$ by our assumption on the noise model for $f = -\rho$.  
Recall that the low-star filtration $F(\K)$ is induced by adding the lower-star of $v_i$s in order where $v_0, v_1, \ldots, v_n$ are sorted in increasing order of $f$-values. Specifically, $F(\K)$ contains the following sequence: 
$\emptyset \subset \bar{K}_1 \subset \bar{K}_2 \subset \cdots \subset \bar{K}_n,$  
where $\bar{K}_i = \bigcup_{j \le i} \LowSt(v_j)$. 
Assume that $e' = <v_a, v_b>$ such that $a < b$ (i.e, $f(v_a) < f(v_b)$). Then $\bar{K}_b$ is the subcomplex in the filtration $F(\K)$ when the edge $e'$ is first included. 
It follows from the persistence algorithm \cite{ELZ02} that, there are two connected components $C_1, C_2 \subset \bar{K}_{b-1}$ that is merged with the addition of $e'$; that is, $v_a \in C_1$ and $v_b \in C_2$. Furthermore, since $e'$ forms a persistence pair  with $u$, this means that $u$ must be the global minimum of one of the components, say $C_1$ w.o.l.g, and $f(u) > f(z)$ where $z$ is the global minimum of the other component $C_2$.  
See Figure \ref{fig:lem4_4}. 

\begin{figure}[H]
\centering
\includegraphics[height=.2\textheight]{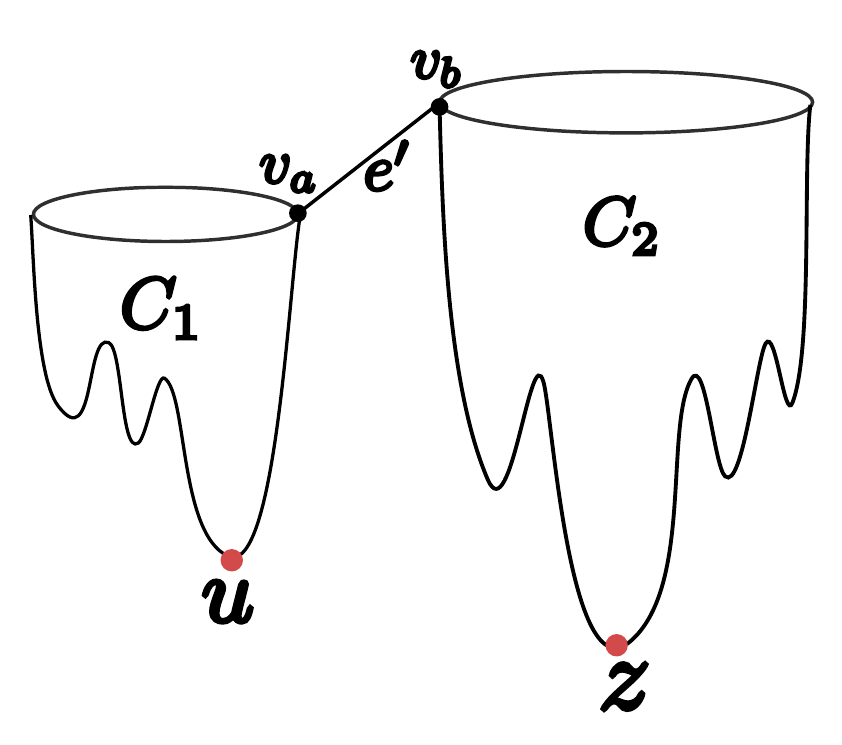}
\caption{Illustration for proof of Proposition \ref{prop:crit-edge-vertex}.}
\label{fig:lem4_4}
\end{figure}

Hence $z \in \thickenG$ as $u \in \thickenG$ and $f(u) > f(z)$. 
This however is not possible. Indeed, note that $\thickenG \subseteq \bar{K}_{b-1}$: This is because that $e' \in \redR$ (and thus $\myf(e') > \delta > \nu$), implying $\bar{K}_{b-1}$ will contain all vertices (thus as well as simplices they span) whose function value is lower than $\delta$, which contains $\thickenG$ (whose vertices have $f$ values $<\nu < \delta$). 
Since both $u$ and $z$ are in $\thickenG$ (which is connected), they are 
already connected in $\bar{K}_{b-1}$, contradictory to that $u\in C_1$ and $z\in C_2$ are two different connected components in $\bar{K}_{b-1}$. Hence the persistence pairing $(u, e)$ cannot exist, and  
there is no critical vertex other than the global minimum $v_0$. 

\subparagraph*{Proof of claim (ii).}
A critical edge forms a persistence pair with either a vertex, or a triangle, or remains unpaired. The last case is not possible as the input domain is simply connected. Any remaining critical edge cannot be paired with a (critical) vertex either by claim (i). 
Statement (ii) then follows. 

\subparagraph*{Proof of claim (iii). }
We now prove Proposition (iii): 
Let $e$ be a critical edge that is considered by \CollectOutputVtwo(). By Proposition~\ref{prop:crit-edge-vertex},
$e$ forms a persistence pair with a triangle $t$. Since $e$ is considered by \CollectOutputVtwo(), $\pers(e,t)> \delta$. Under our noise model, this means that $e$ and $t$ cannot be both in $\thickenG$ or both in its complement $\overline\thickenG$. Then, we have two possibilities:
(i) $e\in\thickenG$ and $t\not\in \thickenG$ or (ii) $t\in\thickenG$ and $e\not\in\thickenG$. In case (i) there is nothing to prove because $e\in\thickenG$; case (ii) is impossible because the function value of $t$ will be less than that of its pairing edge $e$.
\subsection{Proof of Proposition \ref{deform-lem}}
\label{appendix:deform-lem}

First, we need the following obvious claim.
\begin{claim}
Let $G$ be any graph and $T\subseteq G$ be a spanning tree. Let $E\subseteq G\setminus T$ be a set of edges. Let $G_{E,T}\subseteq G$ be the subgraph where $G_{E,T}= (T\cup E)\setminus H$ where $H$ is the largest set of edges in $T$ whose deletions do not eliminate any cycle. Given $T$ and $E$, the set $H$ and hence $G_{E,T}$ is unique.
\end{claim}
We call $G_{E,T}$ in the above claim to be the minimal subgraph with respect to the edge set $E$ and the spanning tree $T$.
Now, consider the spanning tree $T$ of the $1$-skeleton of $\K$ consisting of all negative edges. Addition of positive edges creates cycle. Let $E$ be the set of all positive edges with persistence more than $\delta$. The graph consisting of edges $T\cup E$ has cycles. Consider the minimal subgraph $G_{E,T}$. The graph $\hat{G}$ computed by the algorithm is the graph $G_{E,T}$ plus a maximal set of edges in $T$ whose deletions do not eliminate any cycle in $\Ghat$. Denoting this set of edges as $H$, one has $\hat{G}=G_{E,T}\cup H$.

On the other hand, the union of the boundaries $\cup_iB_i$ of the regions $R_i$ consist of only negative edges or positive edges that are critical. This is because otherwise the edges have to be positive and non-critical, but then the two regions containing such edges are already merged, eliminating them from boundaries. Therefore, $\cup_iB_i$ can be formed by taking the spanning tree $T$ consisting of negative edges, adding the positive edges $E$ to it and then eliminating all edges that cannot eliminate any cycle. Therefore, $\cup_i B_i$
also forms a subgraph of the $1$-skeleton of $\K$ that is minimal with respect to the same set of positive edges $E$ and the spanning tree $T$. Hence, by claim $G_{E,T}=\cup_i B_i$ and thus $\hat{G}=\cup_iB_i \bigcup H$.

\subsection{Proof of Proposition \ref{modchain-prop}}
\label{appendix:modchain-prop}

First we show that the chain $D_i$ constructed by the persistence algorithm~\cite{ELZ02} is a subchain of $C_i^{i}$.
The chain $D_i$ itself is constructed by adding chains as the algorithm searches for the edge $e_i$ to be paired with
$t_i$. Let $D_{i_1}=t_i, D_{i_2},\ldots, D_{i_k}$ be the ordered sequence of chains that are added to construct $D_i$.

We use double induction. First, assume inductively that, for $j < i$, $D_j$ is a subchain of $C_j^j$. It is true
initially for $j=1$ because $C_1^1=D_1=t_1$. To prove the hypothesis for $C_i^i$, assume in a nested induction
that $D_{i_1},\ldots, D_{i_{k-1}}$ are
subchains of $C_i^i$. Initially, 
$t_i$ and hence $D_{i_1}$ is a subchain of $C_i^i$. 
For the nested induction consider the edge $e_{i_k}$ which is in the boundary $\partial (D_{i_1}+\cdots + D_{i_{k-1}})$. The
persistence pair $(e_{i_k},t_{i_k})$ appears before the pair $(e_i,t_i)$. Therefore, the chain $C_{i_k}^{i_{k}}$ is added to
$C_i^{i_k}$ and thus becomes a subchain in $C_i^{i}$. But, the chain $C_{i_k}^{i_{k}}$ contains $D_{i_k}$ as a subchain which
is also added to $C_i^{i}$ as a result. This establishes that $D_i$ is a subchain of $C_i^{i}$.

If the edge $e_i\in \partial D_i$ is not in the boundary $\partial C_i^{i}$, it must be the case that, for some $i'< i$,
$C_{i'}^{i'}$ has been added to $C_i^{i}$ where $\partial C_{i'}^{i'}$ contains $e_i$. In that case, $(e_i,t_{i'})$ must be
a persistence pair according to the construction of $C_i^j$s. This is impossible because $(e_i,t_i)$ is a persistence pair and
$i> i'$. 

To show that $e_i$ is the most recent edge in $\partial C_i^{i}$, assume inductively that $e_j$ is the most recent
edge in $\partial C_j^{j}$ for $j< i$. If $e_i$ is not the most recent edge in $\partial C_i^{i}$, let $e_{i'}$ be 
the most recent one. Let $C_{j'}^{j'}$ be the chain that was added to $C_i^{i}$ because of which $e_{i'}$ was included
in the boundary $C_i^{i}$. Clearly, $j' < i$. Then, $e_{i'}$ was in the boundary $\partial C_{j'}^{j'}$ and by inductive hypothesis
$i' \leq j'$. It follows that $i' \leq j' < i$ reaching a contradiction that $e_i$ is not the most recent edge.
The proposition then follows. 

\subsection{Proof of Proposition \ref{main-prop}}
\label{appendix:main-prop}

	 We induct on $j$. When $j=1$, we have $R_i^*=t_i$ for every $i=1,\ldots, n$. Then, (a) \& (b) are satisfied trivially. Assume that they hold for $j$. Since $R_i^*=C_i^{j}$ by inductive hypothesis, Proposition~\ref{modchain-prop} ensures that $\partial R_j^*$ contains $e_j$, the edge with which $t_j$ pairs with. Then, by inductive hypothesis (b), there is a unique gradient path $\pi_j$ from $t_j$ to $e_j$. The algorithm \SimpVecF() only reverses $\pi_j$. 
	 
	 The set $R_j^*$ is a constituent of the sets $\{R_i^*\},i=j,\ldots,n$ that partition the set of triangles in $\K$. Hence, the edge $e_j\in \partial R_j^*$ necessarily belongs to another
	 boundary $\partial R_{j'}^*$ for some $j'>j$. Observe that since each edge can be incident to at most two triangles, there is a unique such $j'$. 
	 We update $R_{j'}^*:=R_{j'}^*+ R_j^*$. Clearly, the sets $R_i^*, i=j+1,\ldots, n$ partitions the set of triangles in $\K$, proving claim (a.ii). We argue that $R_{j'}^*=C_{j'}^{j+1}$ completing the proof that (a) holds after \SimpVecF() processes $t_j$. Before the update it holds inductively that $R_{j'}^*=C_{j'}^{j}$.
	 After \SimpVecF() processes $(e_j,t_j)$, the only $2$-chain that gets updated is $R_{j'}^*$ because $e_j$ is only in $\partial R_{j'}^*$ where $j'\not = j$. The new $2$-chain $R_{j'}^*$ after the update exactly satisfies the definition of $C_{j'}^{j+1}$ because $R_{j'}^*:=R_{j'}^*+R_j^*=C_{j'}^{j}+C_j^{j}$. This proves claim (a.i).

	 To show that claim (b) holds as well, we need to consider the only updated $2$-chain $R_{j'}^*$. We say that an edge is in a $2$-chain if one of its triangles contains the edge in its boundary. Let $e$ be any edge in $ R_{j'}^*$. If $e$ is in $R_{j'}^*$ before the update, then we already have a gradient path from $t_{j'}$ to $e$ by inductive hypothesis. If $e$ is in $R_j^*$ but not in $R_{j'}^*$ before the update, we have a gradient path from $t_{j'}$ to $e$ in the updated $R_{j'}^*$. This gradient path is obtained by concatenating three sequences, say $\pi_1,\pi_2$ and $\pi_3$. The sequence $\pi_1$ is the gradient path from $t_{j'}$ to $e_j$ in $R_{j'}^*$ before the update. Let $t$ be the last triangle shared by a path from $t_j$ to $e$ and the unique path from $t_j$ to $e_j$ before the update. The sequence $\pi_2$ is the subsequence of the reversed path $\pi_j$ from $e_j$ to $t$ with $e_j$ and $t$ removed, and the sequence $\pi_3$ is the gradient path from $t$ to $e$ that existed in $R_j^*$. This establishes that there is a gradient path from $t_{j'}$ to all edges in $R_{j'}^*$ after the update. Now, to prove (b.i), if $e$ is a boundary edge of $R_{j'}^*$ after the update, it must be the case that $e$ is in the boundary of either $R_j^*$ or $R_{j'}^*$ but not in both before the update. In that case, uniqueness of the gradient path from either $t_j$ or $t_{j'}$ to $e$ before the update implies the uniqueness of the path after the update. Hence, (b.i) holds for updated $R_{j'}^*$ and hence for all $i>j$. 
     
     To prove claim (b.ii), suppose $e \in R^*_{j'}$ has two gradient paths from $t_{j'}$ to it after the update. If $e$ is from $R^*_{j'}$ before the update, then the gradient path from $t_{j'}$ to it will not change. Hence it must be negative in this case. 
     So now suppose $e$ is from $R_j^*$ but not from $R_{j'}^*$ before the update. By the argument in the previous paragraph, we now that new gradient path from $t_{j'}$ to $e$ consists of three portions, and it is easy to see that we can have two paths from $t_{j'}$ to $e$ only if there were two paths from $t_j$ to $e$ in $R_j^*$ before the update. Hence by induction hypothesis, $e$ must be negative as well. This proves claim (b.ii), and finishes the proof of Proposition \ref{main-prop}.

\section{Experiments} 
\label{appendix:sec:exp}
In this section, we perform our algorithm on 2D and 3D density fields. 
The 2D density fields are generated from the GPS trajectories and the goal is to extract the hidden road network behind \cite{WWL15}. 
There are three 3D density fields: A synthetic dataset where the ground truth is known; and two real-life datasets where the noise model assumptions may or may not be satisfied: Specifically, we have 
the \emph{Enzo dataset} \cite{2007arXiv0705.1556N,2004astro.ph..3044O}, which comes from the simulations of cosmological structure formation in university and where the goal is to extract the filament structure behind; and the \emph{Bone dataset} which are Micro CT images of bones from the CT Dataset Archive from CIBC \cite{ct_bone}. 
For the \emph{Bone dataset}, we crop a portion of one bone since the triangulation is huge.
The sizes of these data sets are in Table \ref{table:data_size}, where ``Athens'', ``Beijing'', and ``Berlin'' represents the three 2D density fields obtained from large collection of noisy GPS traces in the three respective cities. 

The input points of the datasets are actually vertices from a 2D or 3D grid, so we obtain a triangulation of input points by simply triangulating each 2D/3D cubic cell. 
It is possible to use a cubical complex directly, but using a triangulation allows us to threshold on the input density function to remove noise and reduce the size of input simplicial complex.  

There are two parameters used in our experiments: the threshold $\delta$ which is the input parameter of \WWLAlg() and \WWLAlgSimp() used for persistence-based simplification; and the parameter $t$ to threshold the density function as used by the thresholding method (not needed for our algorithm). Both parameters are chosen empirically.

Below, we first show the efficiency of our simplified algorithm \WWLAlgSimp() versus the original algorithm \WWLAlg(). 
Then, we present some experimental results showing that, while the practical data does not fall under our noise model, empirically the algorithm still works well (as also demonstrated earlier in work such as \cite{2011MNRAS,WWL15}). Furthermore, as we mentioned earlier, thresholding may be able to produce a graph with theoretical guarantee for input density fields under our noise model, when it is combined with say, medial axis type approaches, to extract the graph after thresholding (even for that it is not yet straightforward to obtain such theoretical result). 
However, we will present experiments that show that simple thresholding does not work well empirically for non-uniformly sampled input. 

\begin{table}[h]
\resizebox{\textwidth}{!}
{
\begin{minipage}{\textwidth}
\centering
\begin{tabular}{ | c | c | c | c| }
  \hline
  Name & \#vertex &\#edge &\#triangle\\ \hline
  Athens & 444,600&1,331,111&886,512\\ \hline
  Beijing &3,754,580&11,255,893 &7,501,314\\ \hline
  Berlin & 80,741&241,084&160,344\\\hline
ENZO & 262,144&1,536,192&2,524,284\\\hline
  Bone & 5,351,976&31,829,419&52,768,394\\
  \hline

\end{tabular}
\end{minipage} }
\caption{Size of the triangulations of the datasets.} \label{table:data_size}
\end{table}

\subparagraph*{Running time.}
We implemented our simplified algorithm \WWLAlgSimp() and now compare its running time with the original algorithm \WWLAlg(). 
As we mentioned in Section \ref{sec:alg} in \cite{newpaper}, this algorithm has already  been simplified so that Morse-cancellation is {\bf only done} for vertex-edge critical pairs. Hence we implemented this improved version of \WWLAlg(), which we refer to as \WWLAlg$^+$() (the version of \SimpVecF() where only vertex-edge critical pairs are canceled is referred to as \SimpVecF$^+$()). Note that a speedup of a factor of at least 2 has already been reported for \WWLAlg$^+$() over \WWLAlg() on GPS data \cite{newpaper}. 
The step of computing the persistence pairing for the negation of density field (both in \WWLAlg$^+$() and our \WWLAlgSimp()) is done by PHAT software package \cite{bauer2017phat}. The comparison of their running time is shown in Table \ref{tab:running_time}. 
Specifically, note the step of computing persistence pairing is common to both the original \WWLAlg() algorithm and our simplified \WWLAlgSimp() version. 
Furthermore, for 3D data, this step is currently the bottleneck (although it may be improved by using persistence algorithm more suitable for volumetric data, such as DiPha). Hence we report the time for this step separately in the 3rd column of the table. The 4th and 5th columns of Table \ref{tab:running_time} show the running time (in seconds) of algorithm \WWLAlg$^+$() (without persistence computation) and of algorithm \WWLAlgSimp() (without persistence computation). 
As we can see, our simplified algorithm is more efficient, generally we see at least a factor of 2 speed-up.

\begin{table}[h]
\resizebox{\textwidth}{!}
{
\begin{minipage}{\textwidth}
  \begin{tabular}{| p{1.1cm}| p{0.6cm} | c | c | c | c| }
  \hline
  Name & $\delta$&Pre-process &\SimpVecF$^+$ +\CollectOutput&\SimpVecFVtwo +\CollectOutputVtwo\\ \hline
  Athens & 0.01&12.3&1.2&0.5\\ \hline
  Beijing & 0.1&97.8 &13.1&5.4\\ \hline
  Berlin & 10&2.0&0.25&0.17\\
  \hline
  \end{tabular}
  \bigbreak
\begin{tabular}{| p{1.1cm} | p{0.6cm} | c | c |c | }
   \hline
  Name & $\delta$&Pre-process &\SimpVecF$^+$ +\CollectOutput&\SimpVecFVtwo +\CollectOutputVtwo\\ \hline
  ENZO &50& 26.5&1.0&0.38\\ \hline
   Bone &40&869 &21.6&8.2\\ \hline
\end{tabular}
\end{minipage} }
\caption{Running time (in seconds) of pre-process (column 3, including filtration and persistence computation ), algorithm \WWLAlg$^+$()  w/o persistence part (column 4) and our simplified algorithm \WWLAlgSimp() w/o persistence part (column 5). 
\label{tab:running_time}}
\end{table}

\subparagraph*{Reconstructions results.}
We do not show the reconstruction results for 2D data sets, as these data sets are originally used in both \cite{WWL15,newpaper}
and our output is the same (Corollary \ref{cor:output}). 
We now show the reconstruction results for 3D data sets. 
Since it is hard to obtain ground truth for real life datasets, we first show the result of a synthetic dataset where we know the ground truth in Figure \ref{fig:arti_data}. 
This dataset is generated as follows: First we create an arbitrary graph (the black lines) 
, then we diffuse this graph by convolution with a Gaussian kernel to obtain a 3D density field. The band width of the Gaussian kernel is 4, comparing to the radius 50 of the input data.
As we can see in Figure \ref{fig:arti_data}, our output graph captures the two loops in the ground truth graph, which follows our theoretical results.

The reconstruction results for real life data sets Enzo and Bone are given in Figure \ref{fig:enzo} and in Figure \ref{fig:bone}, respectively, overlayed with the original density field.  There is no single good threshold $t$ as the density has a rather non-homogeneous distribution, and structures can exist at different level of thresholds (which we will show more shortly below). Hence we provide a volume rendering of the input density field overlapped with our output reconstruction. 

We also provide reconstructions of the bone dataset at different threshold $\delta$ level in Figure \ref{fig:bone_simp}. As we increase the threshold $\delta$, we capture fewer features of the data.

\begin{figure}[H]
\captionsetup[subfigure]{justification=centering}
    \centering
        \begin{subfigure}[b]{0.28\textwidth}
        \includegraphics[height=0.19\textheight]{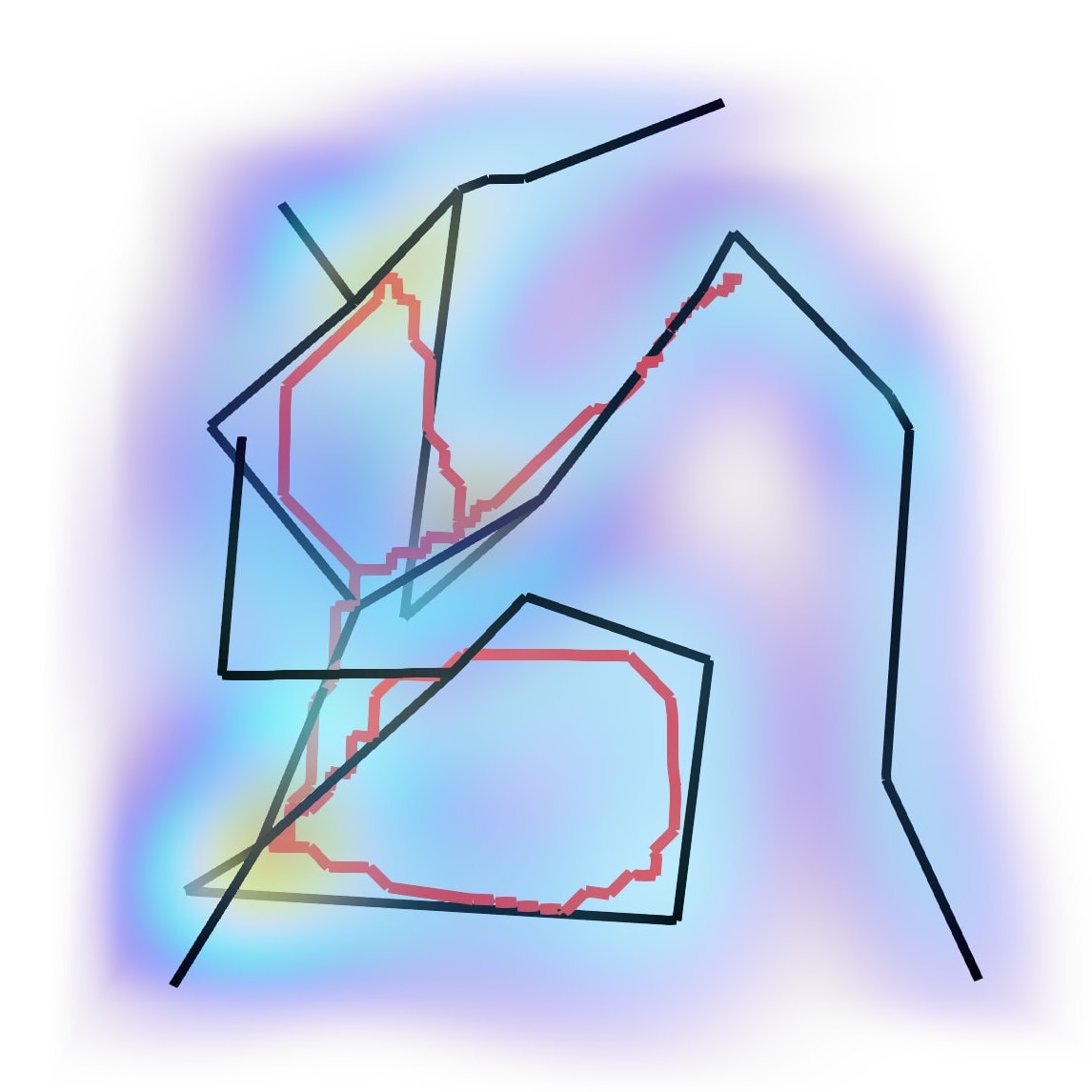}
        \caption{}
        \label{fig:arti_data}
    \end{subfigure}\qquad
    \begin{subfigure}[b]{0.28\textwidth}
        \includegraphics[height=0.19\textheight]{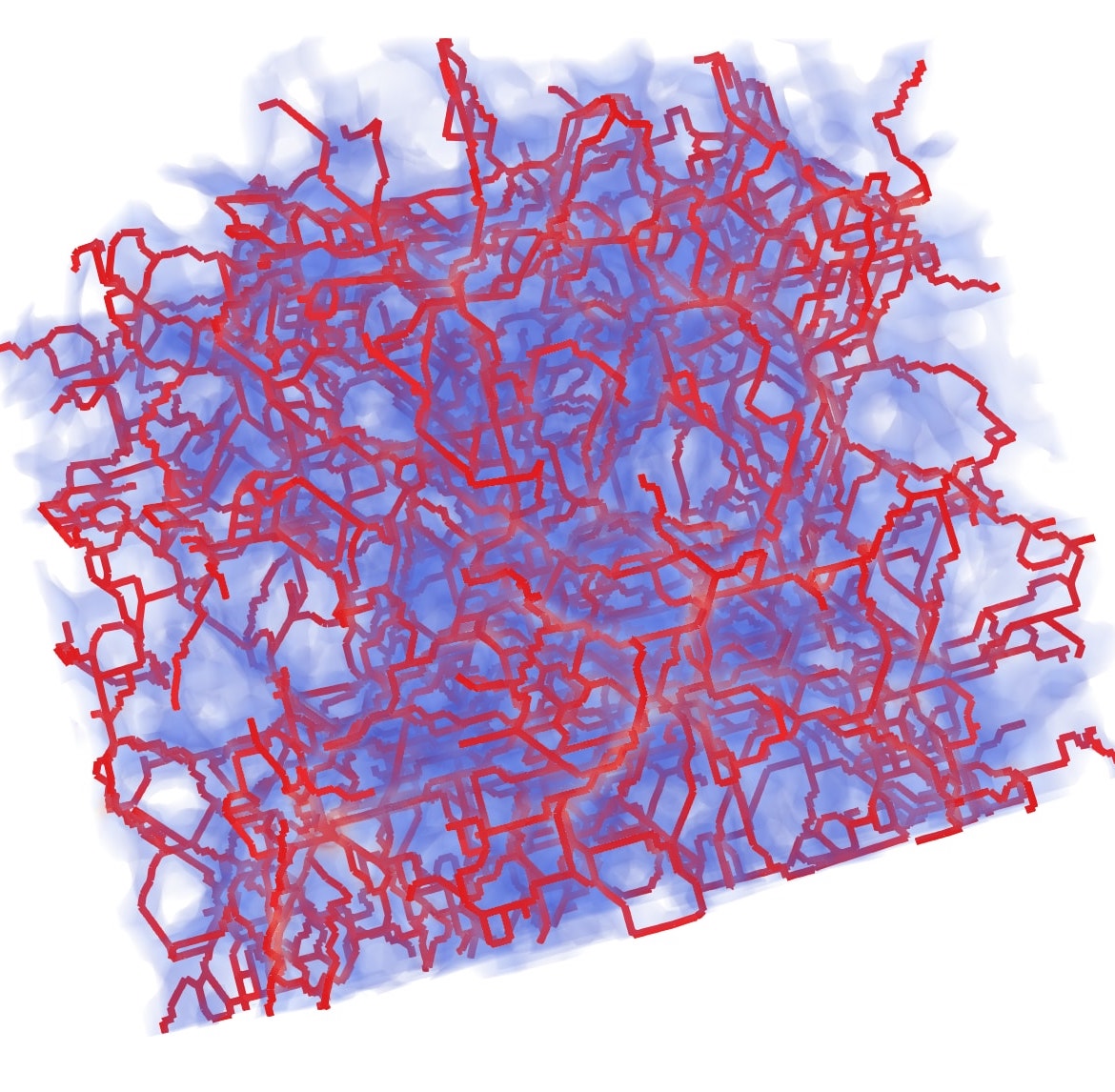}
        \caption{}
        \label{fig:enzo}
    \end{subfigure}\quad
    \begin{subfigure}[b]{0.28\textwidth}
        \includegraphics[height=0.19\textheight]{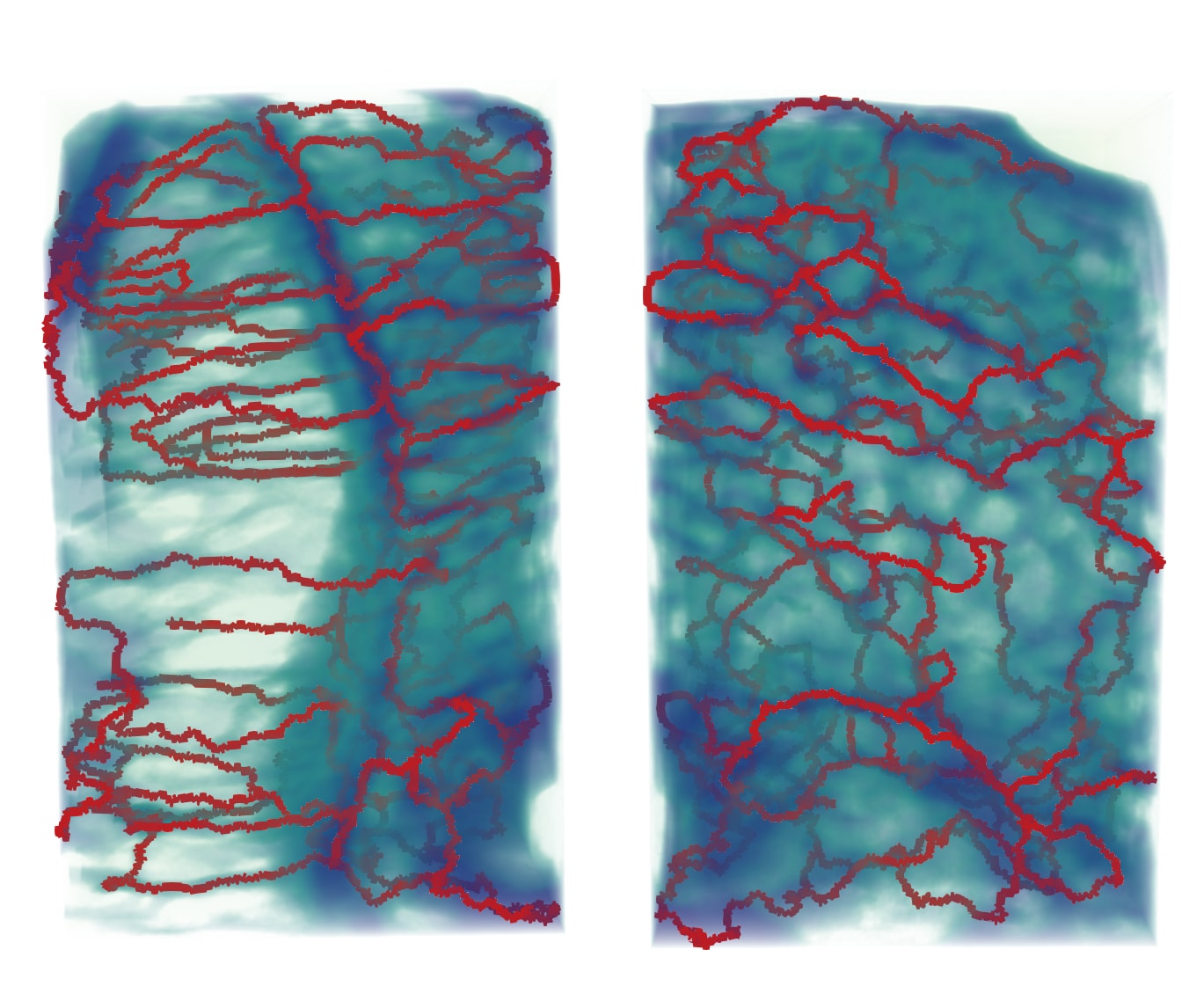}
        \caption{}
        \label{fig:bone}
    \end{subfigure}

 \vspace*{-0.1in}   \caption{In the three figures above, the red lines are the reconstructions, the volumes show the volume renderings of the input density functions. (a) Synthetic dataset.The black lines are the ground truth.  (b) Enzo dataset. (c) Part of a bone dataset.  }
\end{figure}

\begin{figure}[H]
\captionsetup[subfigure]{justification=centering}
    \centering
    \begin{subfigure}[b]{0.22\textwidth}
        \includegraphics[width=\textwidth]{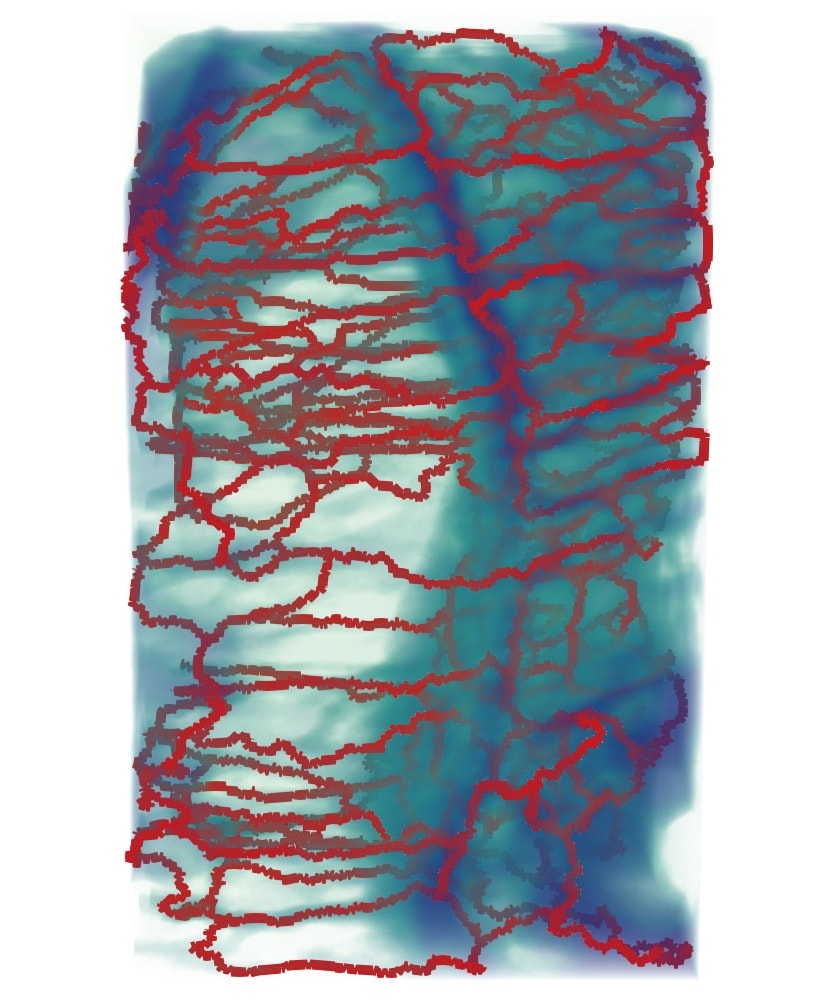}
        \caption{40}
        \label{}
    \end{subfigure}
    \begin{subfigure}[b]{0.22\textwidth}
        \includegraphics[width=\textwidth]{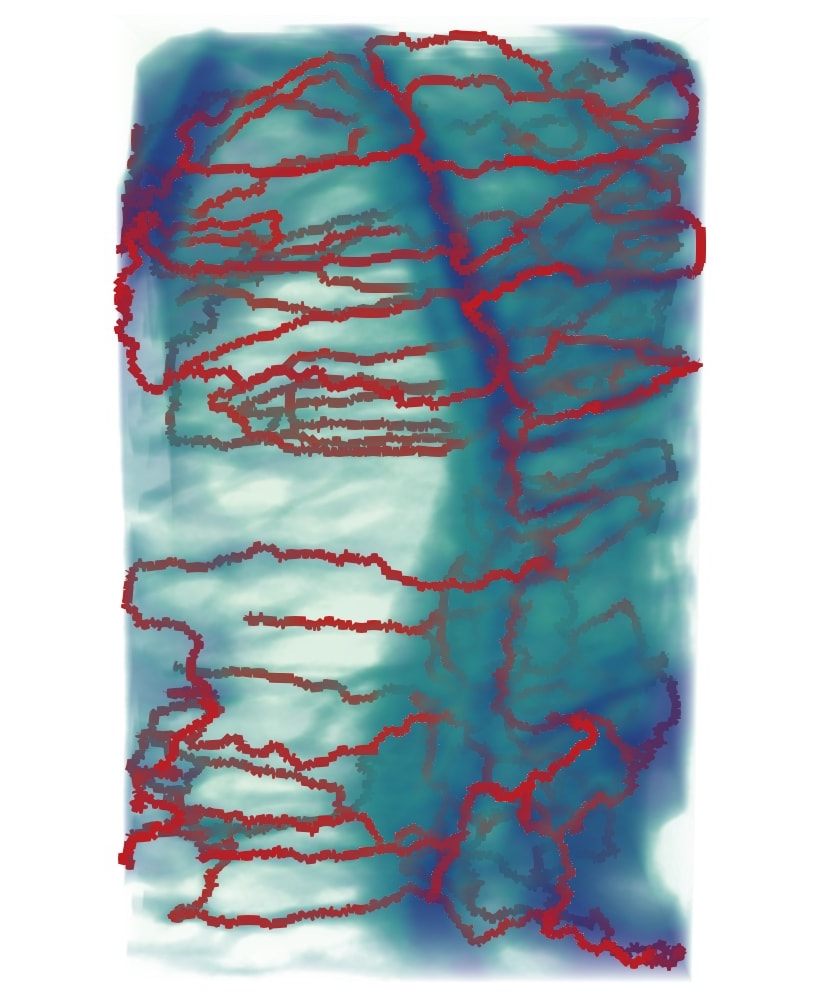}
        \caption{50}
        \label{}
    \end{subfigure}
    \begin{subfigure}[b]{0.22\textwidth}
         \includegraphics[width=\textwidth]{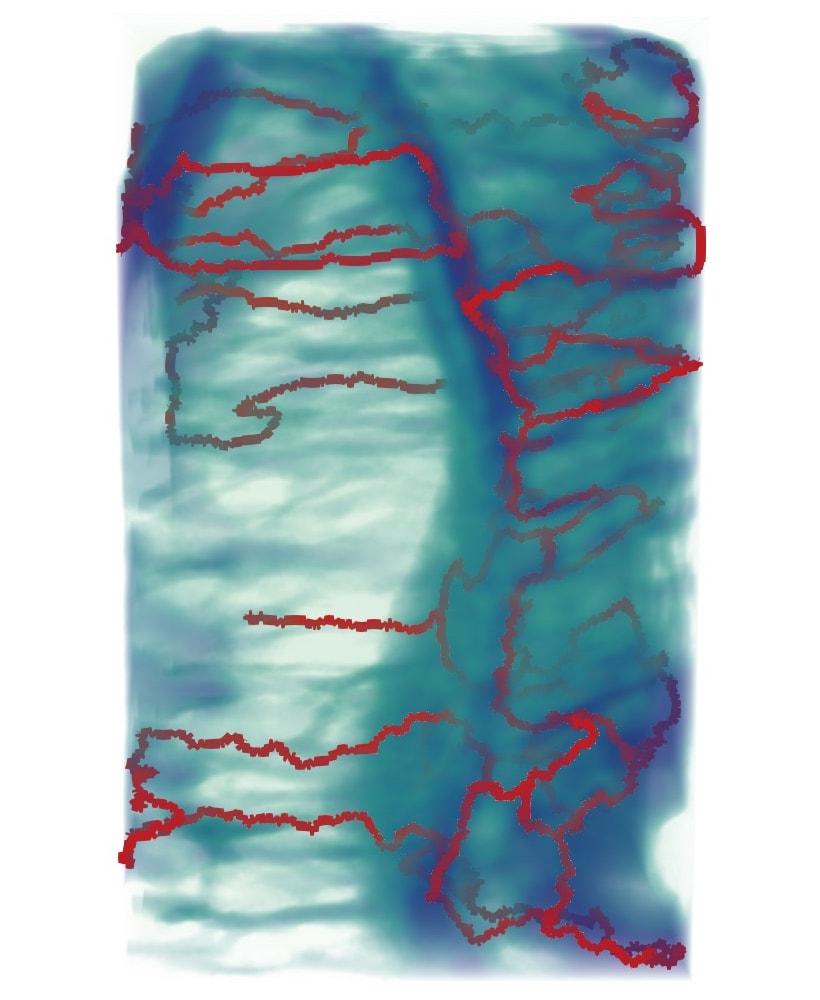}
        \caption{60}
        \label{}
    \end{subfigure}
        \begin{subfigure}[b]{0.22\textwidth}
         \includegraphics[width=\textwidth]{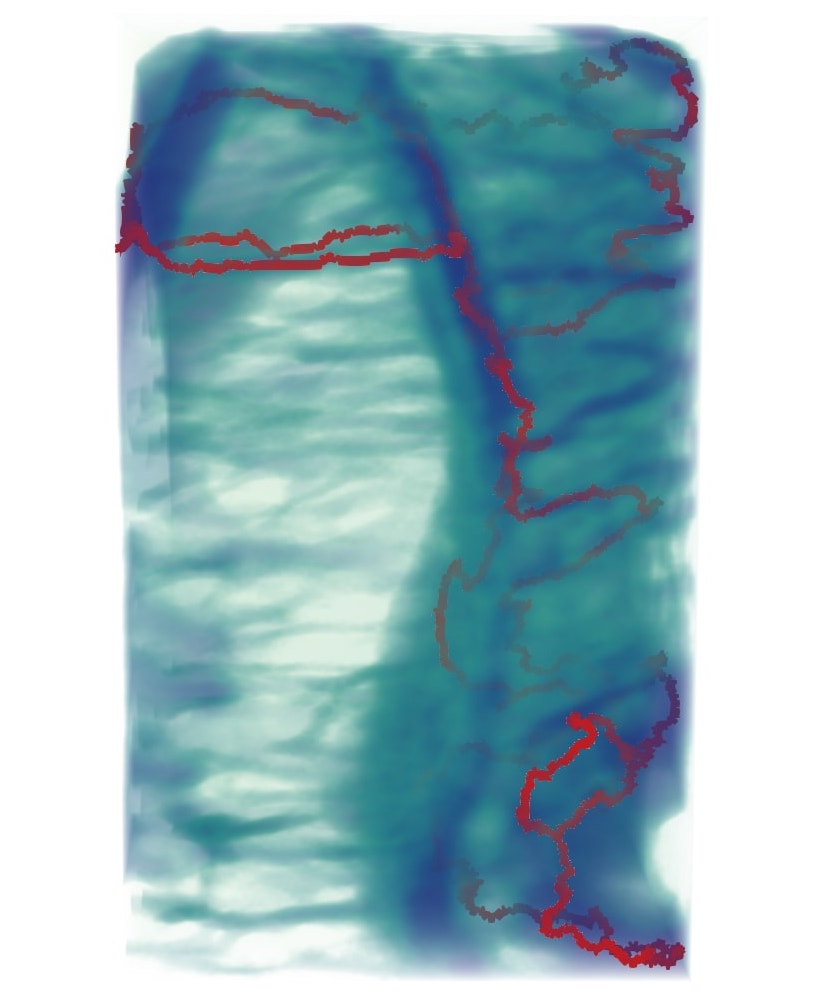}
        \caption{70}
        \label{}
    \end{subfigure}
 \vspace*{-0.1in}   \caption{(a)-(d): Reconstructions of the bone dataset at different threshold $\delta$ level (show in the sub caption).}
  \label{fig:bone_simp}
\end{figure}

\subparagraph*{Comparing with the thresholding method.}
A natural way to denoise is simply thresholding. However, for many data, due to the non-homogeneousness, no single threshold $t$ can capture all features. 
For example, as we show in Figures \ref{fig:enzo_th} and \ref{fig:bone_th}, branching / loop features appear at different time as we vary the density threshold $t$. 
Furthermore, features that appear at high density threshold $t$ may actually be destroyed at low threshold $t$. Hence there is no single good threshold $t$ to capture all features. 
As an example, see the two loops circled with blue in Figure \ref{fig:27} for a high density threshold $t$, which is filled in a lower threshold $t$ in \ref{fig:15} and \ref{fig:10} respectively. However, we need to lower the threshold $t$ as many new features, such as the loops circled with green shown in \ref{fig:15} only appears in a lower threshold $t$. 
Similarly, for the bone dataset, we note that most of the vertical fibers in the lower part can only be captured at a lower threshold $t$. However, at that point, the features in the high density regions (say the top part) are already merged. 

On the other hand, while we do not yet have theoretical guarantees for the discrete Morse based graph reconstruction algorithm for such non-homogeneous data sets, we note that it performs very well empirically, captures these features of different density scales. (The red graphs in both figures are the output reconstruction by our algorithm \WWLAlgSimp.)

\begin{figure}[H]
\captionsetup[subfigure]{justification=centering}
    \centering
    \begin{subfigure}[b]{0.22\textwidth}
        \includegraphics[width=\textwidth]{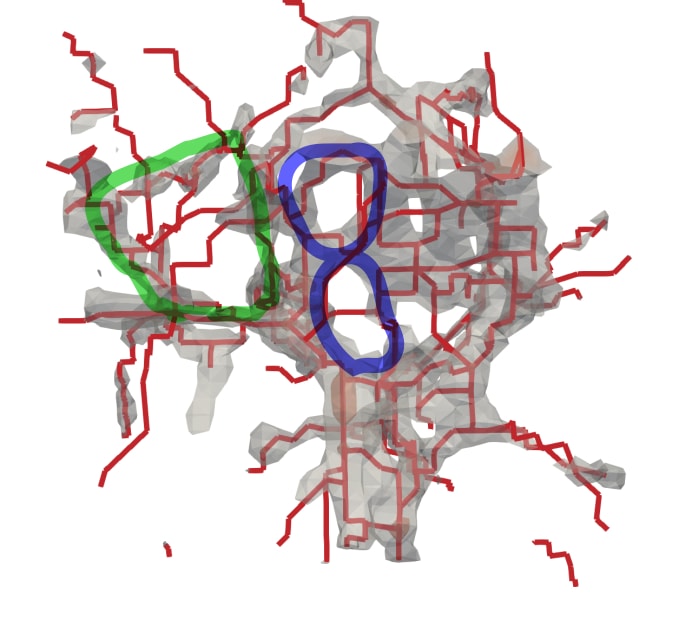}
        \caption{ $2.7\mathrm{e}{-30}$  }
        \label{fig:27}
    \end{subfigure}
        \begin{subfigure}[b]{0.22\textwidth}
        \includegraphics[width=\textwidth]{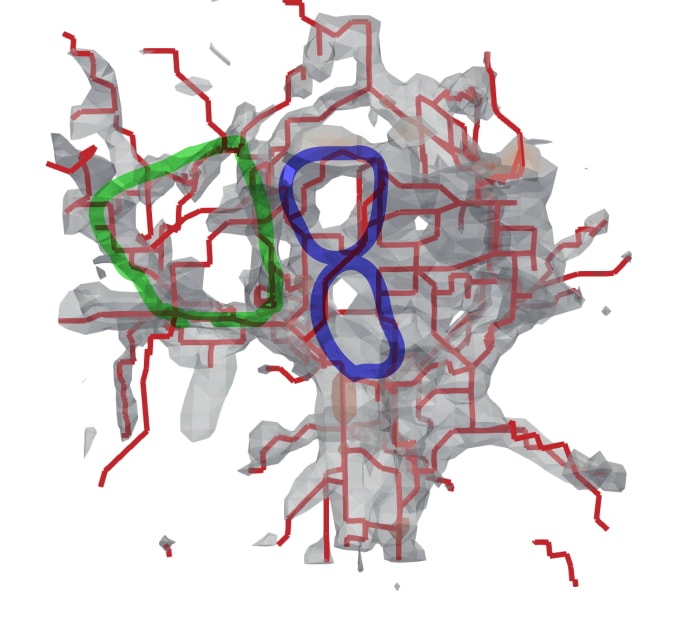}
        \caption{ $2.0\mathrm{e}{-30}$  }
        \label{fig:20}
    \end{subfigure}
        \begin{subfigure}[b]{0.22\textwidth}
        \includegraphics[width=\textwidth]{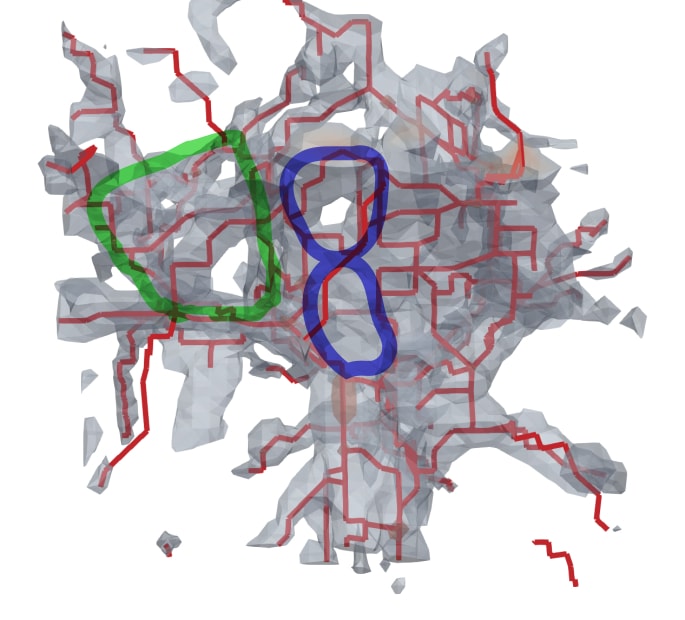}
        \caption{ $1.5\mathrm{e}{-30}$  }
        \label{fig:15}
    \end{subfigure}
        \begin{subfigure}[b]{0.22\textwidth}
        \includegraphics[width=\textwidth]{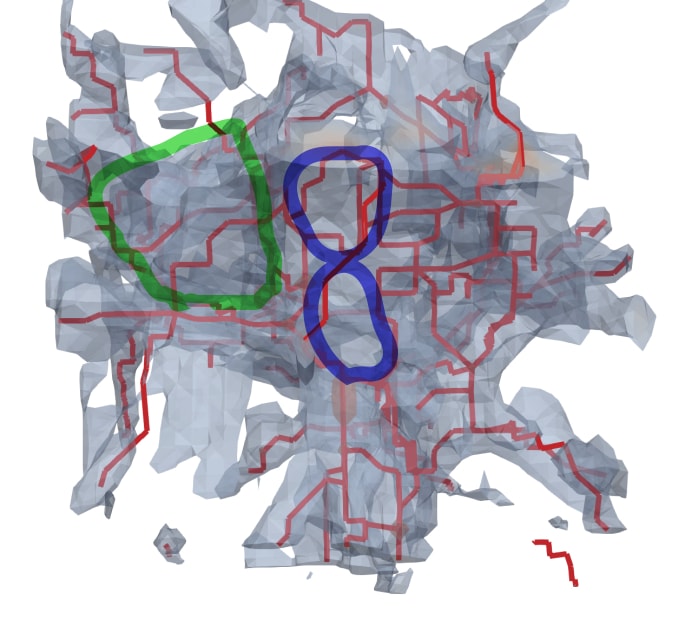}
        \caption{ $1.0\mathrm{e}{-30}$  }
        \label{fig:10}
    \end{subfigure}
\
 \vspace*{-0.1in}   \caption{(a)-(d): Part of the Enzo dataset. The gray volumes are isovolumes with increasing lower bounds (show in the sub caption), the red lines are the reconstructions. Reliable features existing in high thresholds got killed when we lower it.}
 \label{fig:enzo_th}
\end{figure}

\begin{figure}[H]
\captionsetup[subfigure]{justification=centering}
    \centering
    \begin{subfigure}[b]{0.22\textwidth}
        \includegraphics[width=\textwidth]{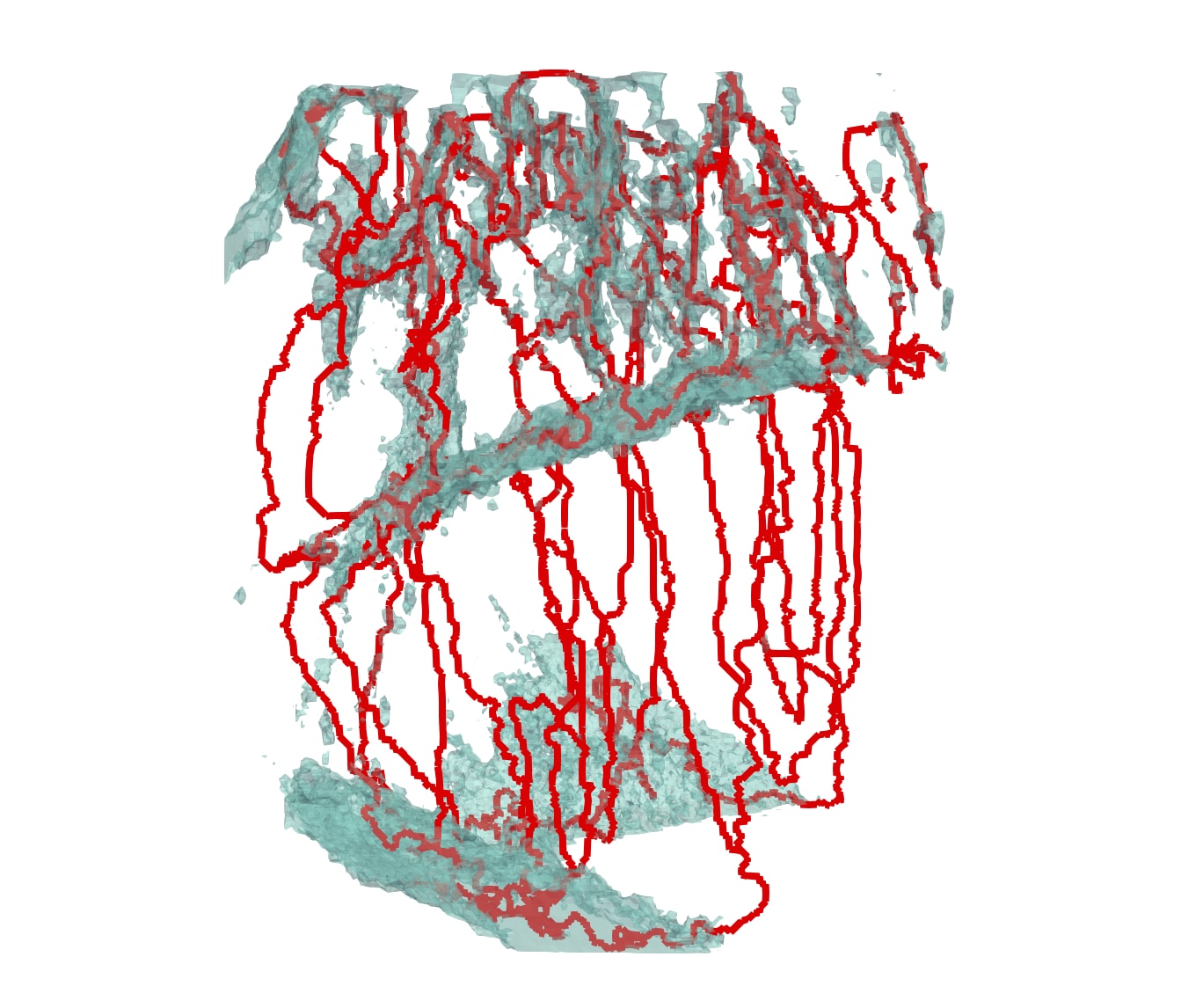}
        \caption{210}
        \label{}
    \end{subfigure}
    \begin{subfigure}[b]{0.22\textwidth}
        \includegraphics[width=\textwidth]{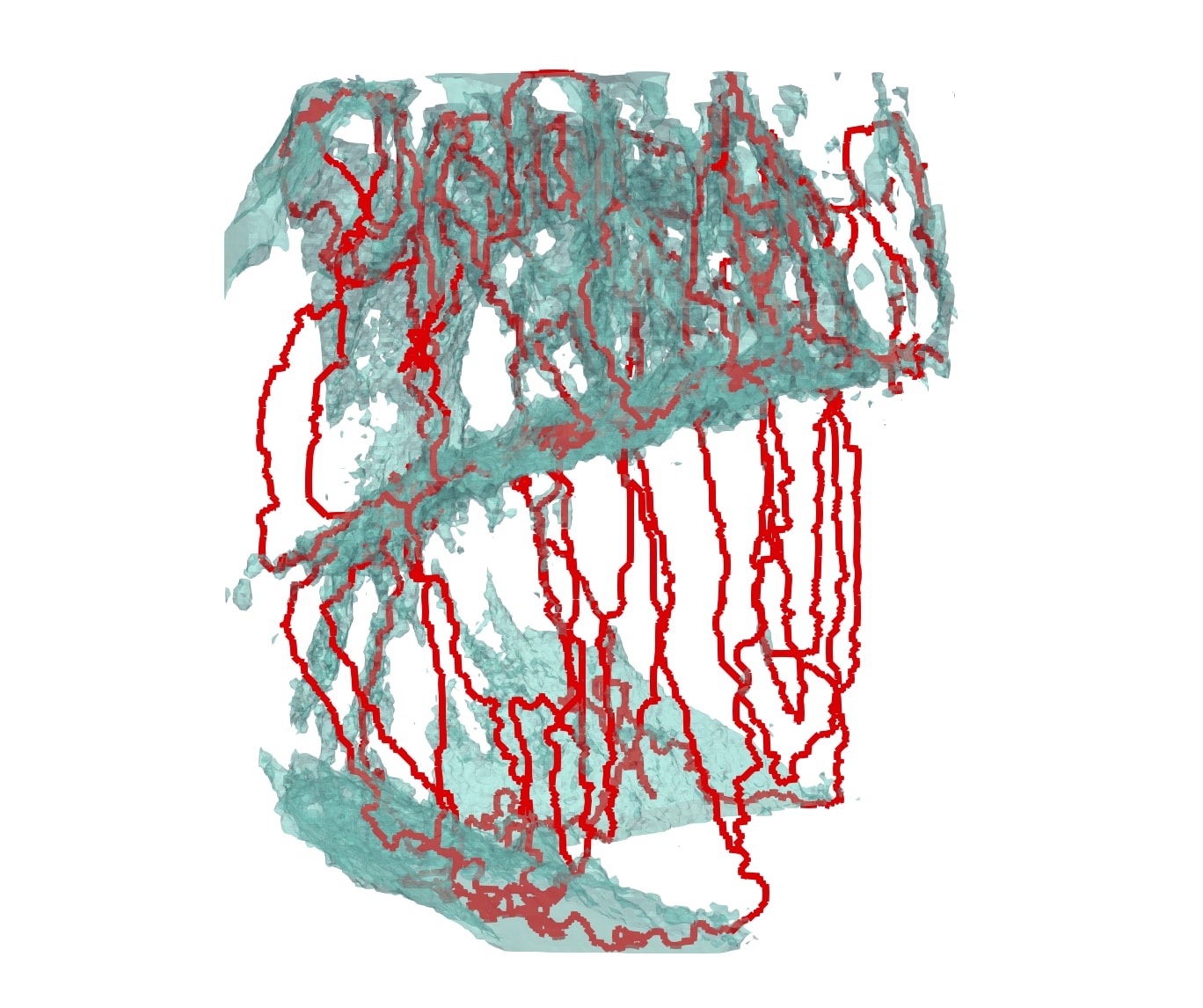}
        \caption{200}
        \label{}
    \end{subfigure}
    \begin{subfigure}[b]{0.22\textwidth}
         \includegraphics[width=\textwidth]{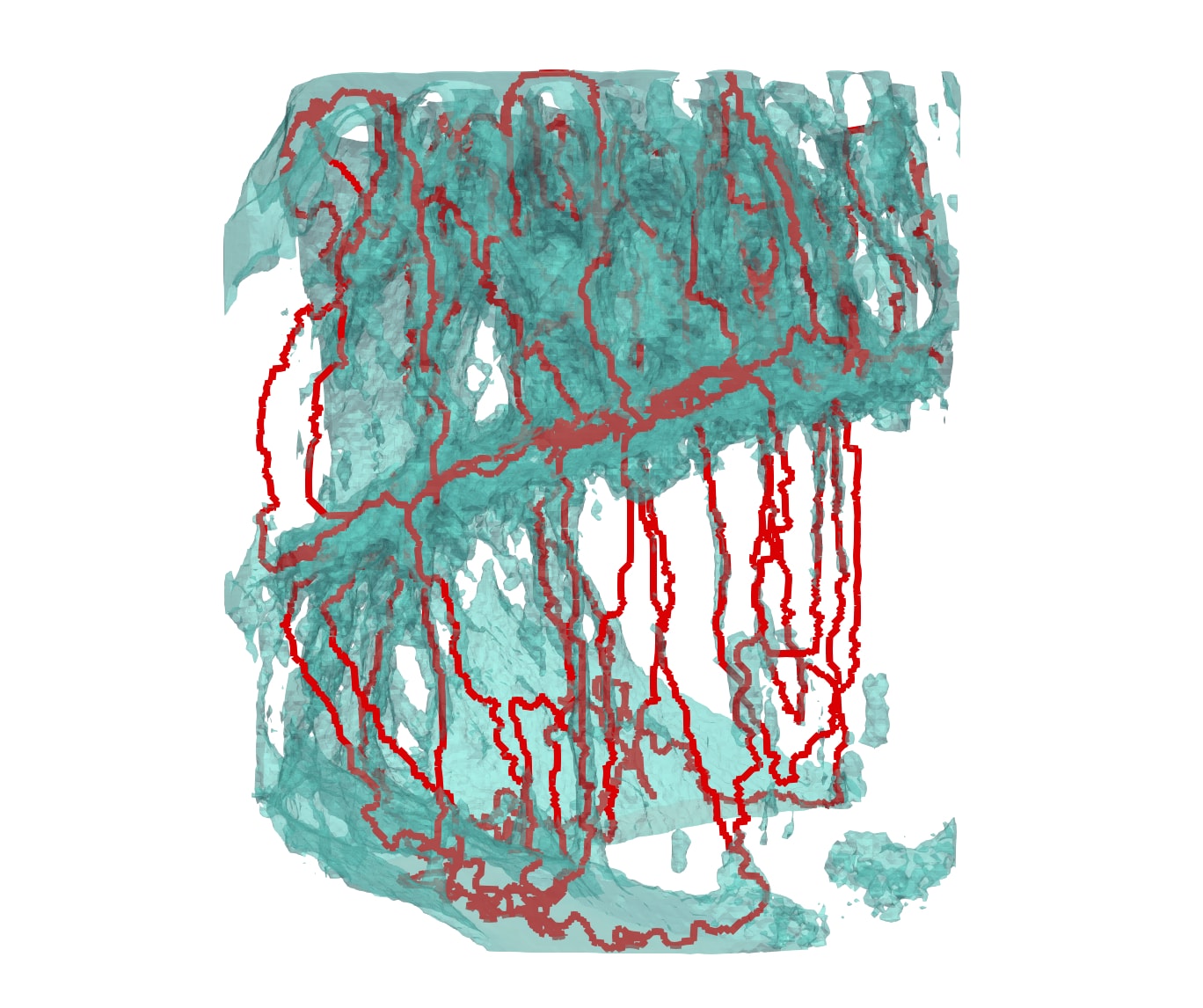}
        \caption{180}
        \label{}
    \end{subfigure}
        \begin{subfigure}[b]{0.22\textwidth}
         \includegraphics[width=\textwidth]{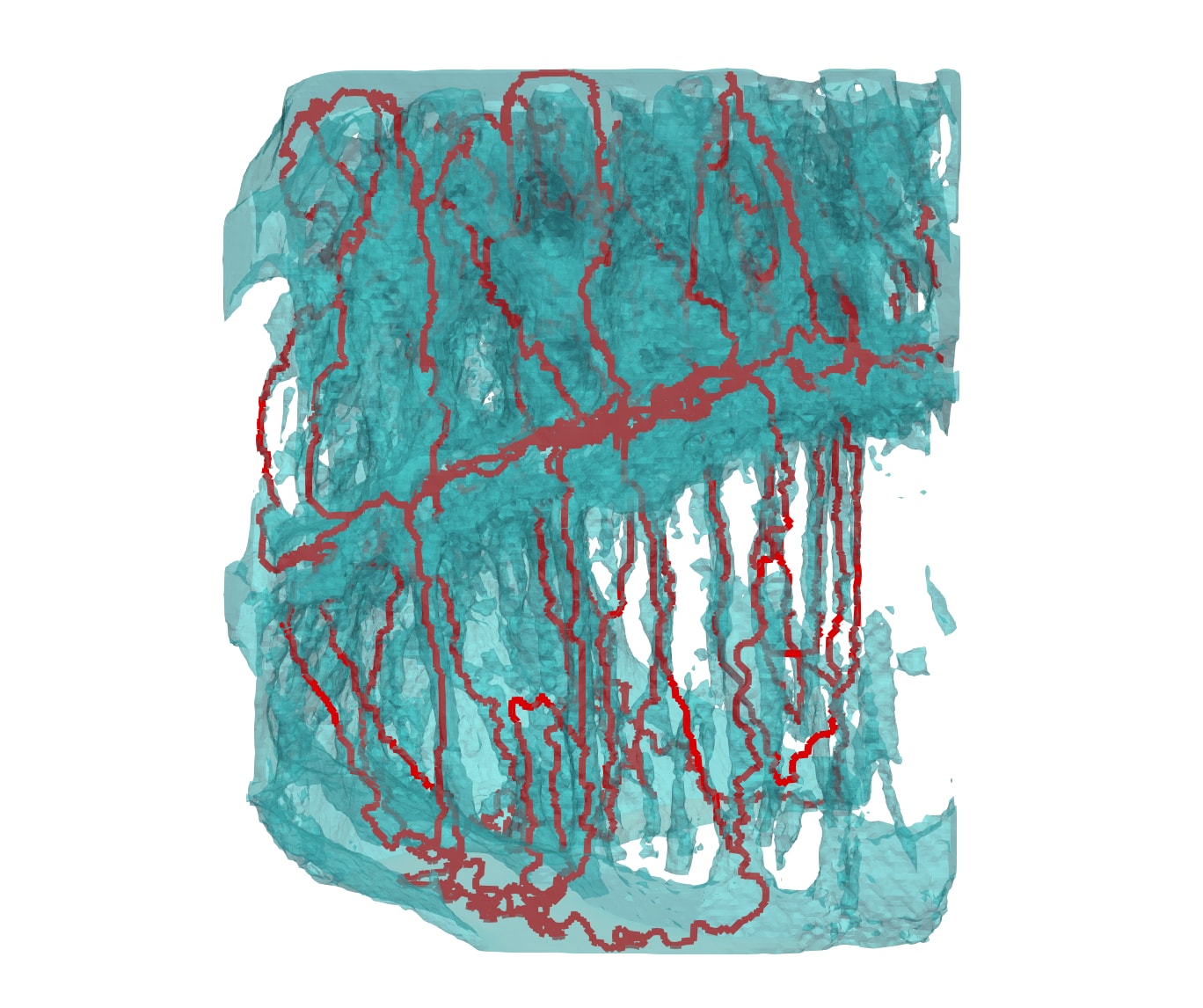}
        \caption{160}
        \label{}
    \end{subfigure}
 \vspace*{-0.1in}   \caption{(a)-(d): Bone dataset. The blue volumes are isovolumes with increasing lower bounds (show in the sub caption), the red lines are the reconstructions.  Low threshold $t$ captures features in lower part, high threshold $t$ captures loops in upper part.}
  \label{fig:bone_th}
\end{figure}

\end{document}